\documentclass[a4paper,UKenglish,cleveref, autoref, thm-restate,arxiv]{lipics-v2021}
\pdfoutput=1 
\hideLIPIcs

\nolinenumbers

\title{Structural Parameterizations for Eternal Vertex Cover} 

\Copyright{Jane Open Access and Joan R. Public} 



\category{} 

\relatedversion{} 

\acknowledgements{}
\author{Neeldhara Misra}{Indian Institute of Technology, Gandhinagar, India \and \url{https://www.neeldhara.com/}}{}{https://orcid.org/0000-0003-1727-5388}{}
\author{Sebastian Ordyniak}{School of Computer Science, University of Leeds, Leeds, UK \and \url{https://algorithms.leeds.ac.uk/profiles/447-2/}}{}{https://orcid.org/0000-0003-1935-651X}{}
\author{Giacomo Paesani}{Sapienza University of Rome, Rome, Italy \and \url{https://giacomopaesani.github.io/}}{}{https://orcid.org/0000-0002-2383-1339}{}
\author{Mateusz Rychlicki}{School of Computer Science, University of Leeds, Leeds, UK}{}{https://orcid.org/0000-0002-8318-2588}{}

\authorrunning{Misra, Ordyniak, Paesani, and Rychlicki}

\usepackage{amsthm}
\usepackage{boxedminipage,graphicx}
\usepackage{charter,parskip,mathpazo}
\makeatletter
\def\thm@space@setup{%
  \thm@preskip=\parskip \thm@postskip=0pt
}
\makeatother
\usepackage[svgnames]{xcolor}

\usepackage{nicefrac}
\usepackage{tikz}
\usetikzlibrary{calc}
\usetikzlibrary{positioning}
\usetikzlibrary{arrows}
\usetikzlibrary{snakes}
\usetikzlibrary{decorations.pathmorphing}
\usetikzlibrary{decorations.markings}
\usetikzlibrary{decorations.pathreplacing}
\usetikzlibrary{shapes.geometric}
\usetikzlibrary{fit,backgrounds}
\tikzset{
    small circles/.style={circle,inner sep=2pt,fill=#1},
    hollow circles/.style n args={2}{circle,inner sep=#1,draw=#2,thick},
    stars/.style={star,inner sep=2pt}
}

\newcommand{\confT}{\ensuremath{{\color{LightSlateBlue}\mathbf{T}}}}
\newcommand{\confR}{\ensuremath{{\color{SeaGreen}\mathbf{R}}}}
\newcommand{\confC}{\ensuremath{{\color{SeaGreen}\mathbf{C}}}}

\newcommand{\confS}{\ensuremath{{\color{Orchid}\mathbf{S}}}}

\newcommand{\EVC}{\textsc{Eternal Vertex Cover}}

\newif\iflong
\newif\ifshort

\longtrue

\iflong
\else
\shorttrue
\fi

\usepackage[noend]{algpseudocode}
\usepackage{algorithm,algorithmicx}

\algnewcommand\algorithmicinput{\textbf{Input:}}
\algnewcommand\INPUT{\item[\algorithmicinput]}

\algnewcommand\algorithmicoutput{\textbf{Output:}}
\algnewcommand\OUTPUT{\item[\algorithmicoutput]}

\begin{document}

\maketitle

\begin{abstract}
\textsc{Eternal Vertex Cover} (EVC) is a turn-based attacker–defender game on an undirected graph $G$. To begin with, the defender places $k$ guards on vertices of $G$. The attacker, on their turn, can choose an edge $e$ not already occupied at both endpoints to ``attack''. The edge $e$ is defended if a guard moves along the edge $e$. The defender, on their turn, can move any subset of guards. A guard can only move to a neighboring vertex. The minimum number of guards needed to indefinitely defend against any sequence of attacks is called the eternal vertex cover number, generalizing the classic vertex cover number. Determining this number is \textsf{NP}-hard in general, motivating the study of parameterized and approximation algorithms. The problem is known to be \textsf{FPT} when parameterized by the cover number, but structural parameters remain relatively unexplored in the literature.

In this work, we explore structural parameterizations for EVC. We show that EVC is \textsf{FPT} parameterized by the cluster vertex deletion number, which generalizes the previously studied parameterization by vertex cover number. We next study the problem parameterized by vertex integrity, which is the smallest number of vertices we need to delete from $G$ so that the resulting graph is a disjoint union of constant-sized components. We first show that Eternal Vertex Cover is \textsf{XP} parameterized by vertex integrity. Then, we develop a polynomial-time approximation algorithm, which computes an additive $6k+1$ ($g(k)$) approximation, where $k$ is equal to the cluster vertex deletion number (vertex integrity). Finally, we show a \textsf{FPT} algorithm for when the deletion set produces ``nice'' connected components, which are components that are bounded in size and satisfy a technical condition.
\end{abstract}

\newpage

\section{Introduction}

The definition of the \emph{eternal vertex cover number} of an
undirected graph $G$ is based on the following $2$-player game between
the \emph{attacker} and the \emph{defender} on $G$. At the start of
the game, the defender places one guard on every vertex of an
arbitrary subset of $V(G)$. Afterwards, the game proceeds in rounds
and in each round the attacker chooses an edge $e$ of $G$ whose endpoints are not both occupied to attack and
the defender can move every guard along an edge such that afterwards
no vertex is occupied by two guards. The attack on $e$ is
\emph{defended} and the current round is \emph{won by the defender} if
the defender moves at least one guard along $e$; otherwise the current
round is \emph{won by the attacker}. The defender wins a (possibly
infinite) play, i.e., sequence of rounds, of the game if he
wins every round of the game and otherwise the attacker wins the
play. Moreover, the minimum number of guards required by the defender
to win the game against every possible play of the attacker is known as the eternal vertex cover number of the graph $G$, denoted $evc(G)$.

\begin{figure}[ht]
\centering
\tikzset{
wnode/.style={shape=circle,draw=black,fill=white,inner sep=0pt,minimum size=4pt},
bnode/.style={shape=circle,draw=black,fill=black,inner sep=0pt,minimum size=4pt}}
\resizebox{\linewidth}{!}{%
\begin{tikzpicture}
\node[bnode] (A1) {};
\node[bnode, right=of A1] (A2) {};
\node[bnode, right=of A2] (A3) {};
\node[bnode, right=of A3] (A4) {};
\node[wnode, right=of A4] (A5) {};
\node[bnode, right=of A5] (A6) {};
\node[bnode, below=of A1] (B1) {};
\node[wnode, right=of B1] (B2) {};
\node[wnode, right=of B2] (B3) {};
\node[bnode, right=of B3] (B4) {};
\node[bnode, right=of B4] (B5) {};
\node[wnode, right=of B5] (B6) {};
\draw(B2)--(A2)--(A1)--(B2)--(B1)--(A1)(B1)--(A2)--(B3)--(A3)--(A4)--(B4)--(B5)--(A4)(A5)--(A6)--(B5);
\draw[thick,color=red,->] (A6)--(B6);
\draw[thick,->] (A4)--(A5);
\begin{scope}[xshift=8cm]
\node[bnode] (A1) {};
\node[bnode, right=of A1] (A2) {};
\node[bnode, right=of A2] (A3) {};
\node[wnode, right=of A3] (A4) {};
\node[bnode, right=of A4] (A5) {};
\node[wnode, right=of A5] (A6) {};
\node[bnode, below=of A1] (B1) {};
\node[wnode, right=of B1] (B2) {};
\node[wnode, right=of B2] (B3) {};
\node[bnode, right=of B3] (B4) {};
\node[bnode, right=of B4] (B5) {};
\node[bnode, right=of B5] (B6) {};
\draw(B2)--(A2)--(A1)--(B2)--(B1)--(A1)(B1)--(A2)--(B3)--(A3)--(A4)--(B4)--(B5)--(A4)--(A5)--(A6)--(B5)(A6)--(B6);
\end{scope}
\end{tikzpicture}}
\caption{Black vertices have a guard each, while white vertices have none. The attacker attacks the left configuration on the red edge, and the defender moves guards in the direction of the arrows. The result is the right configuration.}
\label{fig:intro}
\end{figure}
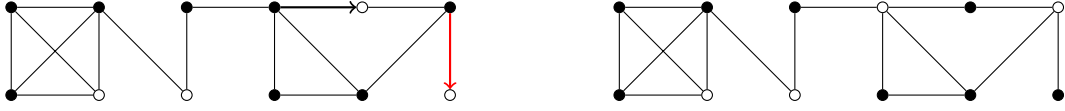

If $S_\ell$ is the subset of vertices that have guards on them after the defender has played her $\ell$-th move, and $S_\ell$ is not a vertex cover of $G$, then the attacker can target any of the uncovered edges to win the game. Therefore, when the defender has a winning strategy, it implies that it can always ``reconfigure'' one vertex cover into another in response to any attack, where the reconfiguration is constrained by the rules of how the guards can move and the requirement that at least one of these guards needs to move along the attacked edge. Therefore, it is clear that $evc(G)\geqslant vc(G)$, where $vc(G)$ denotes the vertex cover number of $G$. It also turns out that twice as many vertices as the $vc(G)$ also suffice to defend against any sequence of attacks. This might be achieved, for example, by placing guards on both endpoints of any maximum matching to begin with and after any attack, reconfiguring the guards to obtain another maximum matching. This also implies $vc(G)\leqslant evc(G) \leqslant 2vc(G)$. A characterization of the graphs for which the upper bound is achieved is known~\cite{KM09}. A characterization for graphs for which the lower bound is achieved remains open, but several special cases have been addressed in the literature~\cite{BCFPRW22}.

The NP-hardness of computing the value of $evc(G)$ was established in~\cite{FGGKS10}, by a reduction from \textsc{Vertex Cover}. Using the ideas in the previous paragraph, the authors obtain a $2$-approximation algorithm for EVC, and also prove that EVC is \textbf{fixed-parameter tractable} (\textsf{FPT}) when parameterized by the number of guards, \(k\). This is achieved by providing a kernelization algorithm
which produces a kernel
whose size is bounded by a function of \(k\).


\newcommand{\cvd}{\textsf{cvd}}
\newcommand{\scd}{\textsf{scd}}

\newcommand{\CVD}{\textsf{CVD}}
\newcommand{\SCD}{\textsf{SCD}}

While the eternal vertex cover number (and analogously dynamic
variants of domination set and related problems) are well studied in
the
literature~\cite{AFI15,BCFPRW22,BKPW22,BMN22,BP22,CC22,CCP25,FGGKS10,KM09,KMC16,MN22,MN23,MN25,PP24}
and
~\cite{BDEMY17,BSL15,BCGMVW04,CCMNP25,FMV15,FGMO18,FMV20,Gal20,GPRT11,KM11,KM12a,KLM16,KM09a,KM12b,LMS19,MINP19,Me17,RS19},
our understanding of these notions still remains in its infancy from
the perspective of structural parameterization and (parameterized)
complexity in general. Indeed, the only known structural parameter for
which the parameterized complexity of eternal vertex cover has been
established thus far is the vertex cover number, where it was shown to
be fixed-parameter tractable~\cite{FGGKS10}. However, attempts to
characterise the parameterized complexity of eternal vertex cover for
treewidth have thus far been far from successful. In particular,
while~\cite{CCP25} provides a linear-time algorithm for a
strongly restricted subclass of series parallel graphs, i.e., graphs of treewidth at most two, it fails to resolve even the full case of series parallel graphs and conjectures that eternal vertex cover will be intractable already on series parallel graphs. This nicely illustrates the challenges with analyzing the parameterized complexity of eternal vertex cover for fairly standard parameterizations. Note also that while eternal vertex cover is known to be NP-hard and in PSPACE~\cite{FGGKS10}, its exact complexity remains open.

In this paper, we consider the parameterized complexity of EVC with
respect to structural parameters beyond vertex cover number, namely,
\emph{cluster vertex deletion number (CVD)}~\cite{DK12} and \emph{small components
  vertex deletion number (SCD)} (a.k.a. vertex integrity, deletion
number to small components, or fracture number)~\cite{BBGLP92}. Given an undirected graph $G$, the CVD of
$G$, denoted by $\cvd(G)$,
is equal to the minimum size of a vertex deletion set $D
\subseteq V(G)$ such that $G-D$ is a disjoint union of
cliques. Moreover, the SCD of $G$, denoted by $\scd(G)$ is equal to the
minimum of $p+|D|$ such that $D$ is a vertex deletion set of $G$ to
components in $G-D$ of size at most $p$. Note that CVD is often
described as a natural
parameter between vertex cover and clique-width and SCD can be seen as a natural parameter
between vertex cover and treedepth. Moreover, both parameters have
been employed as a
structural parameter for a wide variety of problems (e.g.,~\cite{BBBR23,DEGKO21,FGHSO23,GHOORRVS17,GKKMV23,KOP15}).

We show the following results:
\begin{itemize}
\item[(1)] EVC is fixed-parameter tractable parameterized by CVD (\Cref{the:evc-fpt-cvd}).
\item[(2)] EVC is in XP parameterized by SCD (\Cref{the:evc-xp-scd}).
\item[(3)] EVC is fixed-parameter tractable parameterized by a
  restricted version of SCD (\Cref{the:evc-fpt-scd}), where every component $C$ in
  the graph $G$ minus the vertex deletion set is
  additionally required to be what we call self-sufficient. That is,
  we say that a (partial) configuration (position) $P_C \subseteq
  V(C)$ for $C$ is \emph{extendable} if there is a winning position
  $P$ for the defender on $G$ such that $V(C)\cap P=P_C$. Moreover, we
  say that $P_C$ is \emph{minimum} if it has minimum size compared to all
  extendable positions of $C$. Then, $C$ is
  \emph{self-sufficient} if the defender can move in one round between
  any two minimum extendable positions of $C$ without exchanging any
  guard with the deletion set provided that there is no attack on $C$.

  Note that this condition applies to, e.g., cliques, cycles, and generally all
  graphs, where the vertex cover number equals the
  eternal vertex cover number, which have, e.g., been studied in~\cite{MN25}. It is also
  important to note that it is not at all clear at first glance
  whether there are graphs where the defender requires more than one round
  to go between optimal winning configurations.
  As one of our
  contributions, we show that such graphs do indeed exist, and using a computer-based search, we
  were able to show that there is only one such example, which is
  illustrated in~\Cref{fig:not-clique}, among all
  graphs with at most $9$ vertices.
\item[(4)] We provide a lower bound and an upper bound for the EVC
  number of any graph with CVD (or SCD) number
  at most $k$ that differ by at most $3k+1$ (or a function $g(k)$). These results are not only a
  crucial ingredient to obtain (1) and (3), but are also highly
  non-trivial for the case of SCD. Moreover, using the known $2$-approximation for
  CVD~\cite{ADFH23} and the known $(k+1)$-approximation for SCD~\cite{DEGKO21},
  the results additionally allow us to obtain polynomial-time approximation algorithms for EVC with an
  additive error of $6\cvd(G)+1$ and of $g(\scd(G)(\scd(G)+1))$, respectively.
\end{itemize}
Our algorithmic results have the following main ingredients. First we
use the restrictions on the instance posed by the parameters to show
that the number of distinct component types is bounded by a function
of the parameter. We then further show that the number of
non-equivalent configurations in the game for each component is also bounded by a
function of the parameter. These two ingredients are fairly
standard techniques (for the considered parameters) and are also already sufficient to show that EVC is in XP
parameterized by either CVD or SCD. In particular, we obtain a compact
representation of all non-equivalent configurations in the game on $G$
in terms of a \emph{configuration vector}
with $d(k)$ dimensions, for some function $d$ of the parameter $k$,
where each dimension provides the number of components of a certain type
that are in a certain configuration; plus one extra dimension that
provides the current configuration of the deletion set. We then
build the configuration graph of all non-equivalent
configurations and then solve the instance with the help of known methods using the configuration
graph.

Obtaining fixed-parameter tractability is, however, much more
challenging. In particular, to achieve this, it is required to further
restrict the number of configuration vectors that have to be
considered. As a first step in this direction we show that the
difference between the lower bound and the upper bound for EVC can be
bounded in terms of some function $\Delta(k)$ of the considered parameters. This
is fairly easy for the case of CVD, where the function is given by $\Delta(k)=k+1$, but becomes surprisingly involved
for the case of SCD, where we develop new methods to obtain lower
bounds and upper bounds for EVC. Establishing such a bound $\Delta(k)$, then
allows us to only
consider configurations, where all but at most $\Delta(k)$ components
are in a \emph{minimum configuration}, which informally refers to a
configuration using a minimum amount of guards in order to defend any
sequence of attacks on edges incident to the component under the additional
assumption that it can be ensured (using, e.g. other components) that
all vertices in the deletion set are and can remain occupied by guards
throughout the entire game. This allows us to focus on the behaviour
of minimum configurations and here we show that if a component type is
self-sufficient, then it suffices to consider only configurations (for
the whole graph $G$) such that at most one minimum configuration of
the component type occurs more than some function $M(k)$ of the
parameter many times. Therefore, if all component types are
self-sufficient, then it is sufficient to consider only configuration
vectors such that:
\begin{itemize}
\item at most $\Delta(k)$ dimensions correspond to non-minimum
  configurations and
\item for every component type there is only one (minimum) configuration that
  can occur arbitrarily often and all other minimum configurations occur
  at most $M(k)$ times.
\end{itemize}
From this, we can bound the number of important configuration vectors
as follows. There are at most $(d(k)+1)^{\Delta(k)}$ choices of
dimensions corresponding to non-minimum configurations and since each
of those dimensions has value at most $\Delta(k)$, we obtain
$(d(k)+1)^{\Delta(k)}(\Delta(k))^{\Delta(k)}$ possibilities to choose
and assign those dimensions. Moreover, for every component type there
are at most $d(k)$ possibilities to choose the unbounded dimension and
at most $(M(k)+1)^{d(k)}$ possibilities to choose the value for the
remaining dimensions. Finally, there are at most $2^k$ possible
configurations of the deletion set and therefore the total number of
important configurations is bounded by a function of $k$ (i.e.,
$(d(k)+1)^{\Delta(k)}(\Delta(k))^{\Delta(k)}d(k)(M(k)+1)^{d(k)}2^k$).

After introducing the required preliminaries in \Cref{sec:prelim}, in~\Cref{sec:cvd} we show that EVC parameterized by \cvd{} is fixed-parameter tractable (\Cref{the:evc-fpt-cvd}) by reducing to the problem parameterized by \scd{}. In \Cref{sec:scd} we first show that EVC is in \textsf{XP} parameterized by \scd{} (\Cref{the:evc-xp-scd}) and then establish the \textsf{FPT} algorithm for EVC parameterized by \scd{} for self-sufficient instances (\Cref{the:evc-fpt-scd}).
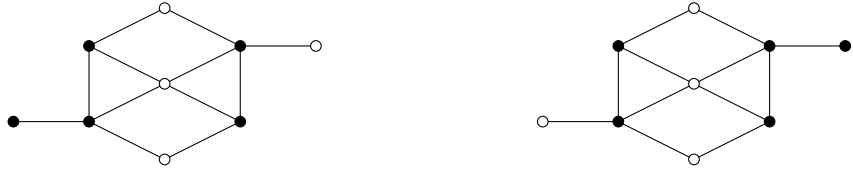
\begin{figure}
\centering
\begin{minipage}{0.5\textwidth}
\centering
\begin{tikzpicture}[scale=1]
\coordinate (A1) at (1,-0.5);
\coordinate (A2) at (1,0.5);
\coordinate (A3) at (0,1);
\coordinate (A4) at (-1,0.5);
\coordinate (A5) at (-1,-0.5);
\coordinate (A6) at (0,-1);
\coordinate (B1) at (-2,-0.5);
\coordinate (B2) at (2,0.5);
\coordinate (C) at (0,0);
\draw \foreach \i in {1,2,4,5}{(A\i)--(C)}
(A5)--(A6)--(A1)--(A2)--(A3)--(A4)--(A5)--(B1)(A2)--(B2);
\draw[fill=white] \foreach \i in {3,6}{(A\i) circle[radius=2pt]}(B2) circle[radius=2pt](C) circle[radius=2pt];
\draw[fill=black] \foreach \i in {1,2,4,5}{(A\i) circle[radius=2pt]}(B1) circle[radius=2pt];
\end{tikzpicture}
\end{minipage}%
\begin{minipage}{0.5\textwidth}
\centering
\begin{tikzpicture}[scale=1]
\coordinate (A1) at (1,-0.5);
\coordinate (A2) at (1,0.5);
\coordinate (A3) at (0,1);
\coordinate (A4) at (-1,0.5);
\coordinate (A5) at (-1,-0.5);
\coordinate (A6) at (0,-1);
\coordinate (B1) at (-2,-0.5);
\coordinate (B2) at (2,0.5);
\coordinate (C) at (0,0);
\draw \foreach \i in {1,2,4,5}{(A\i)--(C)}(A5)--(A6)--(A1)--(A2)--(A3)--(A4)--(A5)--(B1)(A2)--(B2);
\draw[fill=white] \foreach \i in {3,6}{(A\i) circle[radius=2pt]}(B1) circle[radius=2pt](C) circle[radius=2pt];
\draw[fill=black] \foreach \i in {1,2,4,5}{(A\i) circle[radius=2pt]}(B2) circle[radius=2pt];
\end{tikzpicture}
\end{minipage}%
\caption{The black vertices highlight two configurations of a (minimum) eternal vertex cover of a graph $G$ that are not adjacent in the configuration graph $\mathfrak{C}_{evc(G)}(G)$.}
\label{fig:not-clique}
\end{figure}



\section{Preliminaries}\label{sec:prelim}

    Let $G=(V,E)$ be an undirected graph.  An \emph{eternal vertex cover} of $G$ with $k$ guards is a placement of $k$ guards on the vertices of $G$ together with a strategy to respond to an infinite sequence of edge-attacks, under the following rules. In each round, an edge $uv \in E$ whose endpoints are not both occupied is attacked. Each guard may either remain on its current vertex or move to an adjacent vertex, with the requirement that at least one guard moves along the attacked edge (from $u$ to $v$ or vice versa).
    After the move, the guards again occupy $k$ vertices. This new set of occupied vertices must form a vertex cover of $G$, so that every edge of $G$ is still incident to at least one guard. If $G$ has an eternal vertex cover with $k$ guards, we say that \emph{$k$ guards can eternally protect $G$}. The \emph{eternal vertex cover number} of $G$, denoted $evc(G)$, is the minimum $k$ for which $G$ can be eternally protected by $k$ guards.
 %

  %

    \paragraph*{Configuration Graph}
    For a graph $G = (V,E)$, a configuration of order $k$ is a placement of $k$ guards on $V$, with the constraint that every vertex has at most one guard. Configurations of order $k$ are therefore simply $k$-sized subsets of $V$, and will be referred to as such. We will typically refer to guards by the vertices they occupy as opposed to having an explicit identity.

    Given a configuration $\confT$, let $\mathrm{Att}_G(\confT)$ denote the set of edges whose endpoints are not both occupied in $\confT$: $\mathrm{Att}_G(\confT)=\{uv\in E(G)\mid \{u,v\}\nsubseteq \confT\}$. These are precisely the edges that the attacker may choose from $\confT$. Given two configurations $\confT$ and $\confR$ and an edge $e \in E$, we say that $\confT$ is $e$-compatible with $\confR$ if, starting from configuration $\confT$, the guards can either remain on the current vertex or move along edges, every guard moves at most once, some guard moves across the edge $e$, and the resulting configuration is $\confR$. Since configurations record occupied vertices rather than guard identities, cyclic relabellings of guards that leave the occupied set unchanged are not treated as separate defenses. Note that this relationship is symmetric, so $\confT$ is $e$-compatible with $\confR$ if and only if $\confR$ is $e$-compatible with $\confT$. Thus we will speak of $\confT$ and $\confR$ being $e$-compatible. Deciding whether $\confT$ and $\confR$ are $e$-compatible can be done in $n^{\mathcal{O}(1)}$ time \cite{FGGKS10}.

    We will use the following path-based view of a one-round move. Let $\mathcal{P}$ be a family of vertex-disjoint paths, each oriented from its first vertex to its last vertex. We say that $\mathcal{P}$ is \emph{applicable} to a configuration $\confT$ if the first vertex of every path lies in $\confT$, the last vertex of every path lies outside $\confT$, and all internal vertices of the paths lie in $\confT$. Applying $\mathcal{P}$ to $\confT$ shifts the guards one step forward along each path; the resulting configuration is
    \[
      \left(\confT \setminus \{s(P) \mid P \in \mathcal{P}\}\right) \cup \{t(P) \mid P \in \mathcal{P}\},
    \]
    where $s(P)$ and $t(P)$ denote the first and last vertex of $P$, respectively.

    \begin{restatable}{observation}{kdisjointpaths}\label{obs:k-disjoint-paths}
      Let $G=(V,E)$ be a graph, let $\confT,\confR\subseteq V$ be configurations of order~$k$, and let $e\in E$. The configurations $\confT$ and $\confR$ are \emph{$e$‑compatible} if and only if there exist at most $k$ vertex–disjoint paths in~$G$ such that, after orienting each path from its endpoint in $\confT\setminus\confR$ to its endpoint in $\confR\setminus\confT$, applying these paths to $\confT$ results in $\confR$, and the edge $e$ belongs to one of these paths. Equivalently, the endpoints of the paths are exactly the vertices of $\confT\triangle\confR$, each path has one endpoint in $\confT\setminus\confR$ and the other in $\confR\setminus\confT$, and all internal vertices lie in $\confT\cap\confR$.
    \end{restatable}
    \begin{proof}
      ($\Rightarrow$)\;
      Assume that $\confT$ and $\confR$ are $e$‑compatible.
      By definition, there is a single defence move that
      takes the system from~$\confT$ to~$\confR$ in one round in which
      \emph{(i)} every guard either stands on the current vertex or traverses exactly one edge, and
      \emph{(ii)} at least one guard crosses~$e$.
      Label the guards so that guard~$i$ starts on $s_i\in\confT$ and ends on
      $t_i\in\confR$ ($i=1,\ldots ,k$).
      Because two guards never occupy the same vertex simultaneously, for each $v\in \confT\cup \confR$, there is at most one value $i=1,\ldots,k$ such that $v=s_i$ and there is at most one value $j=1,\ldots,k$ such that $v=t_j$. Let $\mathcal{P}$ be the set of edges of the form $s_it_i$. If any two paths $P$ and $P'$ of $\mathcal{P}$ share an endpoint, remove them from $\mathcal{P}$ and add their concatenation. By exhaustive applications of the above rule, we can assume no two paths in $\mathcal{P}$ share an endpoint. Thus, one endpoint of the paths in $\mathcal{P}$ belongs to $\confT\setminus\confR$, the other belongs to $\confR\setminus\confT$, and the internal vertices belong to $\confT\cap \confR$, by definition. Moreover, the paths of $\mathcal{P}$ are vertex-disjoint due to the exhaustive repetition of the concatenation rule and the construction of the edges $s_it_i$. Finally, since by construction $e=s_it_i$ for some $i=1,\ldots,k$, then $e$ is an edge of a path in $\mathcal{P}$, as requested.

      ($\Leftarrow$)\;
      Conversely, suppose there is a set $\mathcal{P}$ of vertex‑disjoint paths from $\confT$ to~$\confR$ with $e$ contained in one path, say~$P_r$, and for every $P\in \mathcal{P}$, the first vertex of $P$ is in $\confT\setminus \confR$ and the last vertex of $P$ is in $\confR\setminus \confT$. Applying $\mathcal{P}$ to $\confT$ yields $\confR$, and since $e$ lies on one of the paths, the corresponding one-round move defends the attack on $e$. Thus $\confT$ and $\confR$ are $e$-compatible.
    \end{proof}

    Given a graph $G$ and a positive integer $k$, the configuration graph of $G$ of order $k$ is denoted $\mathfrak{C}_k(G)$ and is defined as follows:

    \begin{itemize}
    \item The vertex set of $\mathfrak{C}_k(G)$ is the set of all ${n \choose k}$ configurations of order $k$.
    \item Consider vertices $u$ and $v$ in $\mathfrak{C}_k(G)$. Let $F \subseteq E$ be the set of all edges $e$ such that $u$ and $v$ are $e$-compatible.
    For every edge $e \in F$, we add an edge between $u$ and $v$ with label $e$.
    \end{itemize}


Note that the configuration graph $\mathfrak{C}_k(G)$ is an edge-labeled multi-graph. We have the following easy observation connecting the configuration graph to the EVC.
\begin{observation}[{\cite[Lemma 4]{FGGKS10}}]\label{obs:conf-equivalence}
  Let $G$ be a graph. Then, EVC of $G$ is at most $k$
  if and only if the configuration graph $\mathfrak{C}_k(G)$ has a non-empty vertex subset $C$ inducing a connected subgraph with the property that, for every configuration $\confT \in C$, the union of the edge labels of edges incident on $\confT$ inside $C$ contains $\mathrm{Att}_G(\confT)$. Moreover, 
  constructing the configuration graph and checking whether it contains such a non-empty subset $C$ can be achieved in time polynomial in the number of vertices of the configuration graph.
\end{observation}
We now make the notion of a strategy precise.

    \paragraph*{Strategies and Winning Strategies}

    An \emph{attacker strategy} for a graph $G$ is a function $f_a: V(\mathfrak{C}_k(G)) \rightarrow E(G)$ such that $f_a(\confT)\in \mathrm{Att}_G(\confT)$ for every configuration $\confT$. Intuitively, $f_a(\confT)$ specifies the next edge to attack when the current configuration is $\confT$. A \emph{defender strategy} for a graph $G$ is a function

    $$f_d: V(\mathfrak{C}_k(G)) \times E(G) \rightarrow V(\mathfrak{C}_k(G)) \cup \{\bot\},$$

    such that if $f_d(\confT, e) = \confR$, then $\confT$ and $\confR$ are $e$-compatible. Intuitively, $f_d(\confT, e)$ specifies the next configuration to move to when the current configuration is $\confT$ and the attackable edge $e$ is attacked. The output $\bot$ indicates accepting defeat, i.e., no valid defense against the attacked edge is available from the current configuration. By convention, we assume that $\bot$ is not isomorphic to any configuration of $\mathfrak{C}_k(G)$.

    Let $\confT$, $f_a$, $f_d$ be a configuration, an attacker and a defender strategy of $G$, respectively. We say that the game is {\it played based on} $f_a$ and $f_d$ if, at every configuration $\confT$ that occurs in the game, (1) the attacker selects the edge $f_a(\confT)$ and (2) the defender chooses the $f_a(\confT)$-compatible configuration $f_d(\confT,f_a(\confT))$. We say that $f_a$ {\it wins against} $f_d$ on $\confT$ in a game that is played based on $f_a$ and $f_d$ if $\bot$ is reached; otherwise, $f_a$ {\it loses against} $f_d$ on $\confT$. Finally, $\confT$ is {\it winning} (for the defender) if there exists a defender strategy $f_d$ such that $f_a$ loses against $f_d$ on $\confT$, for every attacker strategy $f_a$. Similarly, $\confT$ is {\it losing} (for the defender) if there exists an attacker strategy $f_a$ such that $f_a$ wins against $f_d$ on $\confT$, for every defender strategy $f_d$.

    \begin{observation}
    A configuration is winning if it is part of an eternal vertex cover of $G$, and it is losing otherwise.
    \end{observation}

\section{Cluster Vertex Deletion}
\label{sec:cvd}
We first consider parameterization by \cvd. Let $G=(V,E)$ be a graph and $D\subseteq V$ such that $G- D$ is a disjoint union of cliques. Without loss of generality we can assume that $G$ is connected. Denote by $\ell$ the size of $D$. Recall that $D$ is a \emph{Cluster Vertex Deletion} set (\cvd{}) of $G$. We observe that we can also assume that $G[D]$ is connected by blowing up the parameter by a factor of at most three, since we can always add at most $2\ell - 2$ additional vertices from $V \setminus D$ to make $G[D]$ connected:

\ifshort\begin{observation}[$\star$]\fi
\iflong\begin{observation}\fi\label{obs:connected-deletion-set-cvd}
  Let $G=(V,E)$ be a connected graph and $D \subseteq V$ be a subset such that $G - D$ is a disjoint union of cliques. If $|D| \leqslant \ell$, then there exists $D^\prime \supseteq D$ such that $G[D^\prime]$ is connected and $|D^\prime| \leqslant 3\ell$.
\end{observation}
\iflong
\begin{proof}
  If $D$ induces a connected subgraph in $G$, we are done. Otherwise, suppose $G[D]$ has $c$ connected components, where $2 \leqslant c \leqslant \ell$. We construct a new deletion set $D'$ that induces a connected subgraph as follows. Observe that there exists a pair of components $C_1$ and $C_2$ in $G[D]$ such that the shortest path connecting them has at most two vertices from outside $D$. Indeed, since $G$ is connected, there must be a clique that contains a neighbor of some vertex of $C_1$ and a neighbor of some vertex in $C_2$. Let these neighbors be $u$ and $v$ respectively. We include $u$ and $v$ in $D$. Note that $G[C_1 \cup C_2 \cup \{u,v\}]$ is connected, since $uv \in E$. We repeat this process until $G[D]$ is connected. Let $D'$ be the new deletion set thus obtained. Then $G[D']$ is connected and $|D'| \leqslant \ell + 2(c-1) \leqslant \ell + 2(\ell - 1) \leqslant 3\ell$.
\end{proof}
\fi


\ifshort\begin{observation}[$\star$]\fi
\iflong\begin{observation}\fi
    \label{obs:cvd-bounds}
    Let $G = (V,E)$ be a connected graph with a set $D \subseteq V$ of size $\ell$ such that $G - D$ is a disjoint union of $a$ cliques $C_1, \ldots, C_a$ whose sizes are $q_1, \ldots, q_a$. Then:

    $$ \sum_{i=1}^a (q_i - 1) \leqslant evc(G) \leqslant 3\ell + \sum_{i=1}^a (q_i - 1) + 1 $$
\end{observation}
\iflong
\begin{proof}
    The lower bound follows from the fact that $evc(G) \geqslant vc(G) \geqslant vc(G - D) \geqslant \sum_{i=1}^a (q_i - 1)$. For the upper bound, let $D^\prime \supseteq D$ be a connected deletion set of $G$ of size at most $3\ell$, whose existence is guaranteed by \Cref{obs:cvd-bounds}. Notice that since $G$ is connected, and for every clique $C$ in $G - D^\prime$, at least one vertex in $C$ has a neighbor in $D^\prime$. Since $G[D^\prime]$ is connected, we can always choose $q_i - 1$ vertices of $C_i \in G - D^\prime$ in such a way that the set of chosen vertices, along with $D^\prime$, is a connected vertex cover. This implies the upper bound, since $evc(G) \leqslant cvc(G) + 1 \leqslant 3\ell + \sum_{i=1}^a (q_i - 1) + 1$ \cite{KM09}, where $cvc(G)$ denotes the connected vertex cover number of the graph $G$.
\end{proof}
\fi


The following corollary now follows because of
\Cref{obs:cvd-bounds} together with the known
$2$-approximation algorithm for CVD~\cite{ADFH23}.
\begin{corollary}
  There is a polynomial-time algorithm that computes the EVC number of
  any graph $G$ to within an additive error of at most $6\cvd(G)+1$.
\end{corollary}


The rest of this section describes the \textsf{FPT}-algorithm for EVC parameterized by \CVD{}, i.e., we will show the following theorem.
\ifshort\begin{theorem}[$\star$]\fi
\iflong\begin{theorem}\fi\label{the:evc-fpt-cvd}
    \EVC{} is fixed-parameter tractable parameterized by \CVD.
\end{theorem}

Let $G = (V,E)$ be a graph with a \CVD{} $D \subseteq V$ of size $\ell$. Towards the algorithm, we begin by guessing the size of $evc(G)$, and let this guess be fixed for the rest of this discussion. Notice that this only adds a factor of $3\ell + 1$ to the running time of the algorithm, by \Cref{obs:cvd-bounds}.
Let $C$ be a clique in $G - D$ and let $v \in V(C)$. The signature of $v$ is defined as $\sigma(v) = N(v) \cap D$. If there are more than $\ell + 1$ vertices in $C$ with the same signature $\sigma$, then define $H$ as $G - \{w\}$, where $w$ is a vertex of $C$ with signature $\sigma$. A configuration $\confT$ of $G$ is called $w$-good if $w \in \confT$. As we will see, it will be enough to only work with $w$-good configurations.
For a $w$-good configuration $\confT$, define the projection of
$\confT$ to $H$ as the restriction of $\confT$ to $H$, i.e, $\confT \setminus \{w\}$. Also, for any configuration $\confR$ of $H$, define the
extension of $\confR$ to $G$ as $\confR \cup \{w\}$. We observe that projections and extensions preserve compatibility in the natural sense.

\begin{observation}\label{obs:active-vertex-bound} Let $H$ be obtained from $G$ as described above. Let $uv \in E(H)$. Suppose $\confT$ and $\confT^\prime$ are $uv$-compatible in $G$ and $w$-good. Let us denote by $C_w$ the set of all vertices that share the same signature as $w$ in $C$. Among all ways of moving from $\confT$ to $\confT^\prime$, consider one that
minimizes the total number of guard movements. In such a minimal transition, at most $\lfloor \nicefrac{\ell}{2} \rfloor + 2$ vertices from $C_w$
are ``active'', i.e, they either lose a guard or gain one.
\end{observation}

\begin{proof}
By \Cref{obs:k-disjoint-paths}, the $uv$-compatibility of $\confT$ and $\confT^\prime$ is witnessed by a set $\mathcal{P}$ of vertex-disjoint paths, where each path goes from a vertex in $\confT \setminus \confT^\prime$ to a vertex in $\confT^\prime \setminus \confT$, all internal vertices lie in $\confT \cap \confT^\prime$, and the edge $uv$ belongs to exactly one path. Among all such witnesses, choose $\mathcal{P}$ to minimize the total number of edges (equivalently, guard movements). A vertex $v \in C_w$ is \emph{active} if it lies on some path in $\mathcal{P}$.

We will use the following shortcut repeatedly. Suppose a path $P \in \mathcal{P}$ contains two vertices $a,b \in C_w$, with $a$ appearing before $b$. If $b$ has a successor $x$ on $P$, then $x \in C \cup D$. If $x \in C$, then $ax \in E(G)$ because $C$ is a clique; if $x \in D$, then $ax \in E(G)$ because $bx \in E(G)$ and $a,b$ have the same signature. Thus the subpath from $a$ to $x$ can be replaced by the single edge $ax$, unless this would delete an endpoint of $P$. The symmetric argument applies from the predecessor of $a$. Hence, by the minimality of $\mathcal{P}$, any path with no endpoint in $C$ contains at most one vertex of $C_w$, and any path with an endpoint in $C$ contributes at most as many vertices of $C_w$ as it has endpoints in $C$.

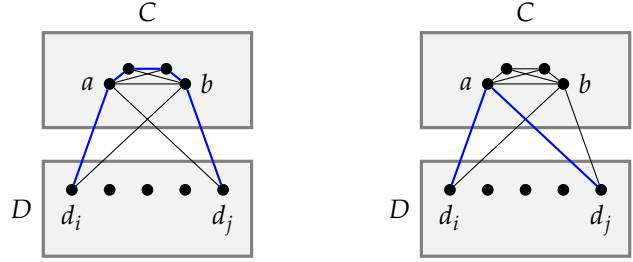
\begin{figure}[ht]
\centering
\tikzset{
wnode/.style={shape=circle,draw=black,fill=white,inner sep=0pt,minimum size=4pt},
enode/.style={shape=circle,inner sep=0pt,minimum size=4pt},
bnode/.style={shape=circle,draw=black,fill=black,inner sep=0pt,minimum size=4pt}}
\begin{tikzpicture}
\node[bnode, label=below:$d_i$] (d) at (-1,-1.2) {};
\node[bnode] (d1) at (-0.5,-1.2) {};
\node[bnode] (d2) at (0,-1.2) {};
\node[bnode] (d3) at (0.5,-1.2) {};
\node[bnode, label=below:$d_j$] (d4) at (1,-1.2) {};
\node[enode] (d') at (1,-1.7) {};
\node[enode]  (w1) at (1,0) {};
\node[enode]  (w2) at (-1,0.5) {};
\node[bnode, label=left:$a$]  (w3) at (-0.5,0.2) {};
\node[bnode]  (w4) at (-0.25,0.4) {};
\node[bnode, label=right:$b$]  (w5) at (0.5,0.2) {};
\node[bnode]  (w6) at (0.25,0.4) {};
\draw[color=blue,thick] (d)--(w3)--(w4)--(w6)--(w5)--(d4);
\draw (w5)--(d)(d4)--(w3)(w4)--(w5)--(w3)--(w6);
\begin{scope}[on background layer]
\node[fill=gray!20!white,draw, rectangle,opacity=0.5,very thick, fit=(d) (d'), inner sep=0.3cm, label=left:$D$] {};
\node[fill=gray!20!white,draw, rectangle,opacity=0.5,very thick, fit=(w1) (w2), inner sep=0.3cm, label=above:$C$] {};
\end{scope}
\begin{scope}[xshift=5cm]
\node[bnode, label=below:$d_i$] (d) at (-1,-1.2) {};
\node[bnode] (d1) at (-0.5,-1.2) {};
\node[bnode] (d2) at (0,-1.2) {};
\node[bnode] (d3) at (0.5,-1.2) {};
\node[bnode, label=below:$d_j$] (d4) at (1,-1.2) {};
\node[enode] (d') at (1,-1.7) {};
\node[enode]  (w1) at (1,0) {};
\node[enode]  (w2) at (-1,0.5) {};
\node[bnode, label=left:$a$]  (w3) at (-0.5,0.2) {};
\node[bnode]  (w4) at (-0.25,0.4) {};
\node[bnode, label=right:$b$]  (w5) at (0.5,0.2) {};
\node[bnode]  (w6) at (0.25,0.4) {};
\draw[color=blue,thick] (d)--(w3)--(d4);
\draw (w3)--(w4)--(w6)--(w5)--(d4)(w5)--(d)(w4)--(w5)--(w3)--(w6);
\end{scope}
\begin{scope}[xshift=5cm,on background layer]
\node[fill=gray!20!white,draw, rectangle,opacity=0.5,very thick, fit=(d) (d'), inner sep=0.3cm, label=left:$D$] {};
\node[fill=gray!20!white,draw, rectangle,opacity=0.5,very thick, fit=(w1) (w2), inner sep=0.3cm, label=above:$C$] {};
\end{scope}
\end{tikzpicture}
\caption{Shortcutting a Type~B path that uses $C_w$ more than once in \Cref{obs:active-vertex-bound}. The blue path (on the left) uses two vertices, namely $a$ and $b$, from $C_w$, and is shortened (on the right) by replacing the subpath between $a$ and $d_j$ with the edge $ad_j$.}
\label{fig:cvd-type-b-shortcut}
\end{figure}

We now partition the paths that contain vertices from $C_w$ into two types.

\textbf{Type~A: paths with at least one endpoint in $C$.} Since $\confT$ and $\confT^\prime$ are both vertex covers and $C$ is a clique, each of $\confT$ and $\confT^\prime$ omits at most one vertex of $C$. Therefore $|V(C) \cap (\confT \triangle \confT^\prime)| \leqslant 2$. Path endpoints lie in $\confT \triangle \confT^\prime$, so the shortcut above implies that Type~A paths account for at most $2$ active vertices from $C_w$ in total.

\textbf{Type~B: paths with no endpoint in $C$ that pass through a vertex of $C_w$.} Since $C$ is a connected component of $G - D$, any path entering $C$ from outside must pass through a vertex of $D$, and similarly for exiting. A Type~B path therefore uses at least two distinct vertices of $D$: one before and one after its visit to $C$.

\begin{figure}[h]
\centering
\tikzset{
wnode/.style={shape=circle,draw=black,fill=white,inner sep=0pt,minimum size=4pt},
enode/.style={shape=circle,inner sep=0pt,minimum size=4pt},
bnode/.style={shape=circle,draw=black,fill=black,inner sep=0pt,minimum size=4pt}}
\begin{tikzpicture}
\node[bnode] (d) at (-1,-1.2) {};
\node[bnode] (d1) at (-0.5,-1.2) {};
\node[bnode] (d2) at (0,-1.2) {};
\node[bnode] (d3) at (0.5,-1.2) {};
\node[bnode] (d4) at (1,-1.2) {};
\node[enode] (d') at (1,-1.7) {};
\node[enode]  (s1) at (-3,0) {};
\node[enode]  (s2) at (-2,1) {};
\node[enode]  (t1) at (3,0) {};
\node[enode]  (t2) at (2,1) {};
\node[enode]  (w1) at (1,0) {};
\node[enode]  (w2) at (-1,1) {};
\node[bnode]  (s3) at (-2.5,0.2) {};
\node[bnode]  (w3) at (-0.5,0.2) {};
\node[bnode]  (w4) at (0,0.2) {};
\node[bnode]  (w5) at (0.5,0.2) {};
\node[wnode]  (w6) at (-0.5,0.7) {};
\node[wnode]  (w7) at (0,0.7) {};
\node[wnode]  (t3) at (2.5,0.2) {};
\draw[color=blue,thick] (w3)--(w6)(d1)--(w4)--(w7);
\draw[color=green!50!black,thick] (s3)--(d) to [out=-30,in=210] (d2)--(w5)--(d4)--(t3);
\begin{scope}[on background layer]
\node[fill=gray!20!white,draw, rectangle,opacity=0.5,very thick, fit=(d) (d'), inner sep=0.3cm, label=left:$D$] {};
\node[fill=gray!20!white,draw, rectangle,opacity=0.5,very thick, fit=(s1) (s2), inner sep=0.3cm, label=above:$C_i$] {};
\node[fill=gray!20!white,draw, rectangle,opacity=0.5,very thick, fit=(t1) (t2), inner sep=0.3cm, label=above:$C_j$] {};
\node[fill=gray!20!white,draw, rectangle,opacity=0.5,very thick, fit=(w1) (w2), inner sep=0.3cm, label=above:$C$] {};
\end{scope}
\end{tikzpicture}
\caption{The paths that contribute active vertices in \Cref{obs:active-vertex-bound}. The green path is Type~B: it has no endpoint in the relevant clique component and enters and leaves it through two vertices of $D$. The blue paths are Type~A: each has an endpoint in the relevant clique component.}
\label{fig:cvd-active-path-types}
\end{figure}
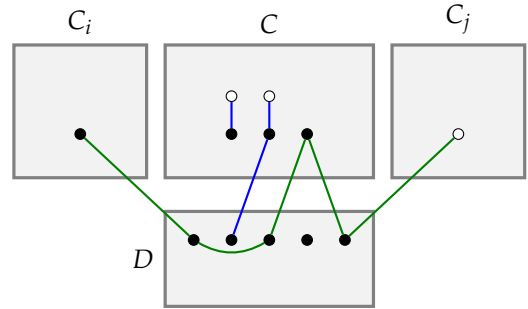

Since the paths in $\mathcal{P}$ are vertex-disjoint, the vertices of $D$ used across different Type~B paths are all distinct, giving at most $\lfloor \nicefrac{\ell}{2} \rfloor$ Type~B paths.

In total, the number of active vertices from $C_w$ is at most $\lfloor \nicefrac{\ell}{2} \rfloor + 2$.
\end{proof}

\ifshort\begin{observation}[$\star$]\fi
\iflong\begin{observation}\fi\label{obs:proj-ext-cvd} Let $H$ be obtained from $G$ as described above. Let $uv \in E(H)$. Then:

\begin{enumerate}
  \item Suppose $\confR$ and $\confR^\prime$ are $uv$-compatible in $H$. Then $\confT := \confR \cup \{w\}$ and $\confT^\prime := \confR^\prime \cup \{w\}$ are $uv$-compatible in $G$ and $w$-good.
  \item Suppose $\confT$ and $\confT^\prime$ are $uv$-compatible in $G$ and $w$-good. Then $\confR := \confT - \{w\}$ and $\confR^\prime := \confT^\prime - \{w\}$ are $uv$-compatible in $H$.
\end{enumerate}
\end{observation}
\iflong
\begin{proof}
The first part follows from the fact that we can simply execute the guard movements that witness the $uv$-compatibility of $\confR$ and $\confR^\prime$ in $H$ in $G$, since $H$ is a subgraph of $G$. None of these movements will involve the vertex $\{w\}$, since $w \notin V(H)$.

For the second part, note that $u \neq w$ and $v \neq w$ since $uv \in E(H)$.
By~\Cref{obs:active-vertex-bound}, in a minimal transition witnessing the $uv$-compatibility of $\confT$ and $\confT^\prime$, at most $\lfloor \nicefrac{\ell}{2} \rfloor + 2$ vertices from $C_w$ are active. If $w$ is not active, then this transition does not involve $w$ at all and can be carried out verbatim in $H$. Otherwise $w$ is active; since $w$ is removed only when $C$ contains more than $\ell + 1$ vertices of signature $\sigma(w)$, we have $|C_w| \geqslant \ell + 2$, and as the number of active vertices is at most $\lfloor \nicefrac{\ell}{2} \rfloor + 2 \leqslant \ell + 1 < |C_w|$, there is a non-active vertex $w' \in C_w$; moreover $w' \neq w$ because $w$ is active. We then replace $w$ with $w'$ on the witness paths: this is valid because $w$ and $w'$ have identical neighborhoods in $D$ and belong to the same clique component $C$. The resulting transition avoids $w$ entirely and can therefore be carried out in $H$.
\end{proof}
\fi

Define a $w$-good strategy as one in which every configuration used is $w$-good.

\ifshort\begin{observation}[$\star$]\fi
\iflong\begin{observation}\fi\label{obs:wgood}
  If there is a winning strategy for $G$ using $k$ guards, then there is a $w$-good winning strategy using $k$ guards.
\end{observation}
\iflong
\begin{proof}
Let $f_d$ be a winning strategy for $G$. We construct a $w$-good
strategy $f_d^*$ as follows. For any edge $e$ not incident to $w$: by
the active vertex bound obtained in~\Cref{obs:active-vertex-bound}, $w$ can remain stationary during any defense, so if the current configuration is $w$-good, the next one is too. For an edge $e = wa$ incident to $w$: if defending $e$ leaves $w$ unoccupied, then since $C$ is a clique and any vertex cover of $G$ must cover all but at most one vertex of $C$, there exists $v \in C$ with $\sigma(v) = \sigma(w)$ that is occupied. Since $v$ and $w$ have identical neighborhoods outside $C$ and are adjacent within $C$, swapping their roles in the defense yields an equivalent $w$-good configuration.
\end{proof}
\fi

\ifshort\begin{lemma}[$\star$]\fi
\iflong\begin{lemma}\fi \label{thm:cvd-bdd}
$evc(G) = evc(H) + 1$.
\end{lemma}
\iflong

\begin{proof}

\textbf{The Forward Direction.} We show $evc(H) \leqslant evc(G) - 1$. By \Cref{obs:wgood}, let $f_d$ be a winning $w$-good strategy for $G$ using $evc(G)$ guards.

Define $f_d^\prime$ for $H$ using $evc(G) - 1$ guards as follows: for any configuration $\confR$ of $H$ and edge $e \in E(H)$, let $f_d^\prime(\confR,e) = f_d(\confR \cup \{w\},e) - \{w\}$. Since $f_d$ is $w$-good, $f_d(\confR \cup \{w\},e)$ contains $w$, so the projection is well-defined and has one guard less. By the second part of \Cref{obs:proj-ext-cvd}, if $\confR \cup \{w\}$ and $f_d(\confR \cup \{w\},e)$ are $e$-compatible in $G$, then $\confR$ and $f_d^\prime(\confR,e)$ are $e$-compatible in $H$. It is easily verified that $f_d^\prime$ is winning, based on the fact that $f_d$ is winning.

\textbf{The Reverse Direction.} We now show that $evc(G) \leqslant evc(H) + 1$. Suppose $f_d$ is a winning strategy for $H$ that uses $evc(H)$ guards. Define $f_d^\prime$ as a partial strategy for $G$ as follows: for any $e \in E(H)$, let $f_d^\prime(\confT,e) = f_d(\confT - \{w\},e) \cup \{w\}$. Note that $f_d^\prime$ is a valid partial strategy because of~\Cref{obs:proj-ext-cvd} combined with the assumption that $e \in E(H)$. Also note that this partial strategy can defend all attacks on edges that are not incident on $w$.

Now suppose some edge $e$ incident on $w$ is attacked, say $e = wy$, when $G$ is in configuration $\confT$. Note that there exists a vertex $x$ with $\sigma(x) = \sigma(w)$ that is occupied in $\confT$, since there are $\ell + 1$ vertices in $C$ whose signature is the same as that of $w$, $\ell + 1 \geqslant 2$, and $C$ is a clique. Therefore, we defend according to $f_d^\prime(\confT,xy)$, but replace $x$ with $w$. We refer to $f_d^\prime(\confT,xy)$ as the original defense and our adaptation of it as the switched defense. Note that the vertex $x$ does not participate in the switched defense at all, since the roles of $x$ and $w$ are switched.

There are two scenarios in terms of how the original defense interacted with the vertex $x$. First, it is possible that the guard on $x$ was moved out of $x$ and some guard was moved into $x$. In the switched defense, $w$ ends up with a guard, and $x$ does not participate in the defense. Thus the resulting configuration is $w$-good and identical to $f_d^\prime(\confT,xy)$. On the other hand, it is possible that the guard on $x$ was moved out of $x$ and no guard was moved into $x$. Thus, at the end of the original defense, $x$ is unoccupied. In this case, at the end of the switched defense, we move the guard on $x$ to $w$. This is a valid move because in the switched defense, $x$ was occupied to begin with, did not participate in the defense, and $w$ was unoccupied at the end. So after this final move, we have again ended up with a configuration that is $w$-good and identical to $f_d^\prime(\confT,xy)$. This concludes the proof.
\end{proof}
\fi

Based on repeated applications of \Cref{thm:cvd-bdd}, we observe that we can start with a graph $G$ and a \CVD{} $D$ and for every signature $\sigma$, remove all but $\ell + 1$ vertices with signature $\sigma$ from any clique $C$ in $G - D$, and be left with an equivalent instance that has the same \CVD{} as the original, but with the guard budget adjusted by the number of vertices removed. This new instance (say $H$) is such that all cliques in $H - D$ have size at most $2^\ell \cdot (\ell + 1)$. Now this instance admits a fixed-parameter tractable algorithm by a reduction to the case of parameterizing by the size of a deletion set into components of bounded size.
\iflong
\begin{proof}[Proof of \Cref{the:evc-fpt-cvd}]
Let $G$ be a graph with a \CVD{} $D$ of size $\ell$. By \Cref{obs:connected-deletion-set-cvd}, we may assume that $G[D]$ is connected, at the cost of increasing $\ell$ by a factor of at most three. There are at most $2^\ell$ distinct signatures $\sigma \subseteq D$. For each clique $C$ in $G - D$ and each signature $\sigma$, if there are more than $\ell + 1$ vertices in $C$ with signature $\sigma$, we repeatedly apply \Cref{thm:cvd-bdd} to remove vertices one at a time until exactly $\ell + 1$ vertices with signature $\sigma$ remain in $C$. Let $w$ denote the total number of vertices removed across all cliques and all signatures. By \Cref{thm:cvd-bdd}, each removal decreases $evc$ by exactly one, so the resulting graph $H$ satisfies $evc(G) = evc(H) + w$.

After this reduction, each clique in $H - D$ has at most $2^\ell \cdot (\ell + 1)$ vertices, since there are at most $2^\ell$ possible signatures and at most $\ell + 1$ vertices per signature. Setting $p := 2^\ell \cdot (\ell + 1)$, we observe that $D$ is also a \SCD{} of $H$ with largest component size at most $p$.

We now argue that every clique $C$ in $H - D$ is self-sufficient in
the sense of \Cref{def:self-sufficient}. Let $s = |V(C)|$.
Because $C$ is a clique any winning configuration of $G$ must contain at least $s-1$
vertices in $C$. Moreover, it is straightforward to verify that every vertex $c \in V(C)$ gives rise to a winning configuration $\confT$ of $G$ with $\confT\cap V(C)=V(C)\setminus \{c\}$, i.e., consider the configuration that contains all vertices of $G$ apart from the the vertex $c$. Therefore, $\{ V(C)\setminus\{c\} \mid c \in V(C)\}$ is the set of minimum configurations of $C$, and since one can easily go between any two of such configurations in one step, it follows that $C$ is self-sufficient.


Since all components of $H - D$ are self-sufficient and have size at most $p$, and the result now follows from~\Cref{the:evc-fpt-scd}.
\end{proof}
\fi


\section{Deletion to Small Components}
\label{sec:scd}

\subsection{A Lower Bound}
To develop a lower bound, we need to examine the components in isolation, but at the same time capture the dynamics of its interaction with the rest of the instance via the deletion set. To this end, we introduce an auxiliary problem that simplifies the outside interaction by empowering the deletion set directly and bringing focus to just one component.

In the {\sc Bootstrapped EVC} problem, we are given a graph $G=(V,E)$, a subset $S\subseteq V$ and a positive integer $t$.  The goal is to determine if there is there a configuration $\confT$ of an eternal vertex cover of $G$ such that $|\confT \cap S|\leq t$?

Note that {\sc Bootstrapped EVC} is a generalization of {\sc EVC} as the two problems coincide when $S=V$. Let $bevc(G,S)=\min \{t~|~(G,S,t)$ is a YES instance for {\sc Bootstrapped EVC}\}. Let $D\subseteq V$ be a vertex set and denote by $\mathcal{D}$ the set of components of $G-D$. Denote the size of a largest component in $\mathcal{D}$ by $p$ and the size of $D$ by $\ell$. Let $C$ be a component in $\mathcal{D}$. We define the graph $G(C)$ that is obtained from $G[D\cup C]$ after adding a set $B$ of pairwise non-adjacent new vertices of size $\ell+p+1$ such that $B$ and $D$ form the two partition classes of a complete bipartite graph. See \Cref{fig:gc} for an example of $G(C)$.

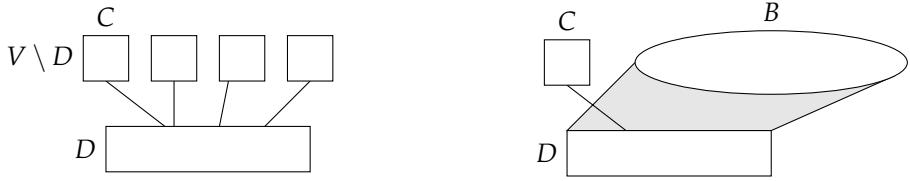
\begin{figure}[h]
\centering
\begin{minipage}{0.5\textwidth}
\begin{tikzpicture}[scale=0.6]
\coordinate (A) at (-5,0);
\coordinate (B) at (-4,1);
\coordinate (C) at (-3.5,0);
\coordinate (D) at (-2.5,1);
\coordinate (E) at (-2,0);
\coordinate (F) at (-1,1);
\coordinate (G) at (-0.5,0);
\coordinate (H) at (0.5,1);
\coordinate (I) at (-4.5,-2);
\coordinate (L) at (0,-1);
\draw (-4.5,0)--(-3.2,-1)( -3,0)--(-3,-1)( -1.8,0)--(-2,-1)(0,0)--(-1,-1);
\draw[fill=white]
(A)rectangle(B)(C)rectangle(D)(E)rectangle(F)(G)rectangle(H)(I)rectangle(L);
\node[left] at (-5,0.5) {$V\setminus D$};
\node[left] at (-4.5,-1.5) {$D$};
\node[above] at (-4.5,1) {$C$};
\end{tikzpicture}
\end{minipage}%
\begin{minipage}{0.5\textwidth}
\begin{tikzpicture}[scale=0.6]
\coordinate (A) at (-5,0);
\coordinate (B) at (-4,1);
\coordinate (C) at (-3.5,0);
\coordinate (D) at (-2.5,1);
\coordinate (E) at (-2,0);
\coordinate (F) at (-1,1);
\coordinate (G) at (-0.5,0);
\coordinate (H) at (0.5,1);
\coordinate (I) at (-4.5,-2);
\coordinate (L) at (0,-1);
\draw[fill=gray!20!white] (-3,0.5)--(2.7,0.2)--(0,-1)--(-4.5,-1)--(-3,0.5)(-4.5,0)--(-3.2,-1);
\draw[fill=white](A)rectangle(B)(I)rectangle(L)(0,0.5)ellipse[x radius=3,y radius=0.7];
\node[above]at (0,1.2) {$B$};
\node[left] at (-4.5,-1.5) {$D$};
\node[above] at (-4.5,1) {$C$};
\end{tikzpicture}
\end{minipage}
\caption{The graph $G$ (left) and $G(C)$ (right) for some component $C$ of $\mathcal{D}$.}
\label{fig:gc}
\end{figure}

A {\it state} $\confR$ for a component $C$ is a subset of vertices of $C$. We say that a component $C\in \mathcal{D}$ is {\it defendable} w.r.t. a state $\confR$, if $D\cup B\cup \confR$ is winning in $G(C)$. Let $\confT$ and $\confT'$ be a configuration of $G$ and $G(C)$, respectively; we say that $\confT'$ {\it dominates} $\confT$ if $\confT\cap (D\cup C)\subseteq \confT'\cap (D\cup C)$.

Let $lb(C)$ be equal to $bevc(G(C),C)$: this value lower bounds the minimum number of guards present in the component $C$ in a configuration of any eternal vertex cover of $G(C)$.

\ifshort\begin{restatable}[$\star$]{lemma}{lbdbound}\fi
\iflong\begin{restatable}{lemma}{lbdbound}\fi\label{obs:lowerbound}
Let $G=(V,E)$ be a graph, $D\subseteq V$, and let
$\mathcal{D}$ be the set of components of $G\setminus D$. Then, $$lb_D(G)=\sum_{C\in \mathcal{D}} lb(C)\leq evc(G).$$
\end{restatable}
\iflong
\begin{proof}
Let $\confT\subseteq V$ be any configuration of $G$ with $|\confT|<lb_D(G)$. Then, there exists a component $C\in \mathcal{D}$, such that the state of $C$, i.e. $\confT\cap V(C)$, has size smaller than $lb(C)$. We show that $\confT$ is losing in $G$.

Since $|\confT\cap V(C)|<lb(C)$, the configuration $\confT'=D\cup B\cup (\confT\cap V(C))$ is losing in $G(C)$. Therefore, there is a winning strategy $f_a'$ for the attacker on $G(C)$ starting from $\confT'$.

It remains to show that $f_a'$ also wins in $G$. We do so by starting from $\confT$ and showing the contrapositive, i.e., we show that if the defender has a winning strategy in $G$ starting from $\confT$, then the defender has a winning strategy in $G(C)$ starting from $\confT'$.


\newcommand{\PPP}{\mathcal{P}}

Let $f_d$ be a winning strategy for the defender on $G$ starting from $\confT$ and define the strategy $f_d'$ for the defender on $G(C)$ as follows.

Let $\confS_G$ be any configuration in $G$ of size $|\confT|$ and let $\confS$ be a configuration of $G(C)$ with $|\confS|=|\confT'|$ and $D\subseteq \confS$ that dominates $\confS_G$. Now consider an attack on an edge $e$ starting from position $\confS$ in $G(C)$.
If $e$ is an edge in $E(G[D\cup B])$, then $f_d'(\confS,e)$ can be set to a configuration isomorphic to $\confS$. Therefore, we can assume that $e$ is incident to at least one vertex in $C$, which implies that $e$ is also an in $G$. Let $\confS_G'=f_d(\confS_G,e)$.
 We will show next that the defender can reach a position $\confS'$ from $\confS$ (while defending $e$) such that $\confS'$ dominates $\confS'_G$ and $D \subseteq \confS'$, and this means we can set $f_d'(\confS,e)=\confS'$.

Let $\PPP_G$ be the set of vertex-disjoint paths that witness that $\confS_G$ and $\confS_G'$ are $e$-compatible and whose existence is guaranteed by \Cref{obs:k-disjoint-paths}. We prove how to obtain a set of vertex-disjoint paths $\PPP$ from $\PPP_G$ that shows that $\confS$ and $\confS'$ are $e$-compatible.
Let $P=(v_1,\ldots,v_t)\in \PPP_G$, for some $t\geq 2$. Without loss
of generality, we can assume that $v_1\in \confS_G\setminus \confS_G'$ and $v_t\in \confS_G'\setminus \confS_G$. Now we do classify $P$ depending on the location of $v_1$ and $v_t$ as follows. The path $P$ is of
\begin{itemize}
\item {\bf Type I} if $v_1\in V(C)$ and $v_t\in V(C)\setminus \confS$;
\item {\bf Type II} if $v_1\in D\cup V(C')$, for some component $C'\in \mathcal{D}$ and $v_t\in V(C)\setminus \confS$;
\item {\bf Type III} if $P$ is not of Type I or II.
\end{itemize}
See also \Cref{fig:types}. Moreover, let $1\leq d\leq d'\leq t$ be the minimum and maximum, respectively, integers such that $v_d$ and $v_{d'}$ belong to $D$.

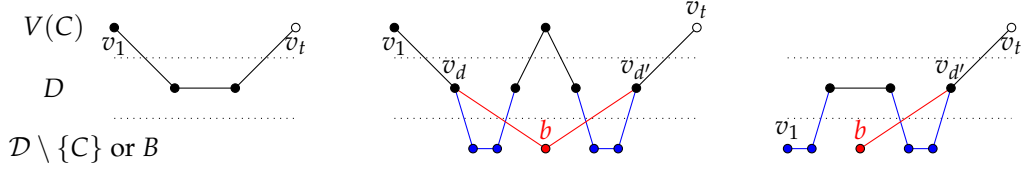
\begin{figure}
\centering
\begin{minipage}[t]{0.3\textwidth}
\centering
\begin{tikzpicture}[scale=0.8, baseline={(0,0)}]
\coordinate (A1) at (-1,1.5);
\coordinate (A2) at (0,0.5);
\coordinate (A3) at (1,0.5);
\coordinate (A4) at (2,1.5);
\coordinate (B) at (2,-0.5);
\coordinate (C2) at (-1,1);
\coordinate (C3) at (-1,0);
\coordinate (D2) at (2,1);
\coordinate (D3) at (2,0);
\draw[dotted] \foreach \i in {2,3}{(C\i)--(D\i)};
\draw (A1)--(A2)--(A3)--(A4);
\draw[fill=black] (A1) circle[radius=2pt](A2) circle[radius=2pt](A3) circle[radius=2pt];
\draw[fill=white] (A4) circle[radius=2pt];
\draw[color=white] (B) circle[radius=2pt];
\node[below] at (A1) {$v_1$};
\node[below] at (A4) {$v_t$};
\node at (-2,1.5) {$V(C)$};
\node at (-2,0.5) {$D$};
\node at (-1.5,-0.5) {$\mathcal{D}\setminus \{C\}$ or $B$};
\end{tikzpicture}
\end{minipage}
\hfill
\begin{minipage}[t]{0.3\textwidth}
\centering
\begin{tikzpicture}[scale=0.8, baseline={(0,0)}]
\coordinate (A1) at (-2,1.5);
\coordinate (A2) at (-1,0.5);
\coordinate (A3) at (0,0.5);
\coordinate (A34) at (0.5,1.5);
\coordinate (A4) at (1,0.5);
\coordinate (A5) at (2,0.5);
\coordinate (A6) at (3,1.5);
\coordinate (B2) at (-0.7,-0.5);
\coordinate (B3) at (-0.3,-0.5);
\coordinate (B4) at (1.3,-0.5);
\coordinate (B5) at (1.7,-0.5);
\coordinate (R) at (0.5,-0.5);
\coordinate (C2) at (-2,1);
\coordinate (C3) at (-2,0);
\coordinate (D2) at (3,1);
\coordinate (D3) at (3,0);
\draw[dotted] \foreach \i in {2,3}{(C\i)--(D\i)};
\draw (A1)--(A2)(A3)--(A34)--(A4)(A5)--(A6);
\draw[color=blue] (A2)--(B2)--(B3)--(A3)(A4)--(B4)--(B5)--(A5);
\draw[color=red] (A2)--(R)--(A5);
\draw[fill=black] \foreach \i in {1,2,3,34,4,5}{(A\i) circle[radius=2pt]};
\draw[fill=blue] \foreach \i in {2,3,4,5}{(B\i) circle[radius=2pt]};
\draw[fill=red] (R) circle[radius=2pt];
\draw[fill=white] (A6) circle[radius=2pt];
\node[below] at (A1) {$v_1$};
\node[above] at (A5) {$v_{d'}$};
\node[above] at (A6) {$v_t$};
\node[above] at (A2) {$v_d$};
\node[color=red,above] at (R) {$b$};
\end{tikzpicture}
\end{minipage}%
\hfill
\begin{minipage}[t]{0.3\textwidth}
\centering
\begin{tikzpicture}[scale=0.8, baseline={(0,0)}]
\coordinate (A34) at (0.5,1.5);
\coordinate (A3) at (0,0.5);
\coordinate (A4) at (1,0.5);
\coordinate (A5) at (2,0.5);
\coordinate (A6) at (3,1.5);
\coordinate (B2) at (-0.7,-0.5);
\coordinate (B3) at (-0.3,-0.5);
\coordinate (B4) at (1.3,-0.5);
\coordinate (B5) at (1.7,-0.5);
\coordinate (R) at (0.5,-0.5);
\coordinate (C2) at (-0.7,1);
\coordinate (C3) at (-0.7,0);
\coordinate (D2) at (3,1);
\coordinate (D3) at (3,0);
\draw[dotted] \foreach \i in {2,3}{(C\i)--(D\i)};
\draw (A3)--(A4)(A5)--(A6);
\draw[color=blue] (B2)--(B3)--(A3)(A4)--(B4)--(B5)--(A5);
\draw[color=red] (R)--(A5);
\draw[fill=black] \foreach \i in {3,4,5}{(A\i) circle[radius=2pt]};
\draw[fill=blue] \foreach \i in {2,3,4,5}{(B\i) circle[radius=2pt]};
\draw[fill=red] (R) circle[radius=2pt];
\draw[fill=white] (A6) circle[radius=2pt];
\node[above] at (B2) {$v_1$};
\node[below] at (A6) {$v_t$};
\node[above] at (A5) {$v_{d'}$};
\node[color=red,above] at (R) {$b$};
\end{tikzpicture}
\end{minipage}
\caption{From left to right, examples of paths of Type I (whose vertex set is contained in $V(C)\cup D$), Type I (whose vertex set intersects $\mathcal{D}\setminus \{C\}$), and Type II, related to the construction of $\PPP$ in \Cref{obs:lowerbound}. In each of these examples, the black, blue and red parts belong to both $G$ and $G(C)$, only to $G$ and only to $G(C)$, respectively. The colour-filled vertices hold a guard, while the white ones are without a guard.}
\label{fig:types}
\end{figure}

For every $P \in \PPP_G$, we add the following paths to $\PPP$:
\begin{itemize}
\item If $P$ is of Type I, then add $P$ to $\PPP$ if $V(P)\subseteq D\cup
  V(C)$ and add the path $v_1\ldots v_dbv_{d'}\ldots v_t$ to $\PPP$ otherwise, for some $b\in B\cap \confS$.
\item If $P$ is of Type II, then add the path $bv_{d'}\ldots v_t$, for some $b\in B\cap \confS$, to $\PPP$.
\item Otherwise, add no path to $\PPP$.
\end{itemize}
Note that $|B\cap \confS|\geq \ell$: indeed, $|B\cap \confS|\geq |\confT'|-|C|-|D|=|B|+|D|+|\confT\cap V(C)|-|C|-|D|=\ell+1+|\confT\cap V(C)|\geq \ell$.
Since $|B\cap \confS|\geq \ell$ and because there are at most $\ell$
paths of $\PPP$ that use a vertex in $D$, there are sufficiently many vertices $b$ in $B\cap \confS$ that can be used to construct the above paths.




By construction, the paths of $\PPP$ are contained in $G(C)$ and are disjoint. Let $\confS'$ be the configuration obtained from $\confS$ by applying the paths in $\PPP$; note that this can be done because, by construction, for every such that $P$, every vertex apart from the last endpoint (denoted above by $v_t$) is not in $\confS$. Note that $D\subseteq \confS'$, because $D \subseteq \confS$ and the paths in $\PPP$ only move out guards from $B\cup V(C)$. Furthermore, $\confS'$ dominates $\confS_G'$ because $\confS$ dominates $\confS_G$ and because $(\confS_G'\setminus \confS)\cap V(C)$ obtain a guard from one of the paths in $\PPP$.
\end{proof}
\fi

\subsection{An Upper Bound}

Our first aim is to define a configuration $\confT^*$ that is winning
for the defender on $G$ and that does not use too many guards compared
to the lower bound provided in the previous subsection. Before doing
so we need to provide a mechanism that will allow us to assign a
certain number of private components for every $d \in D$; these
components will be used to supply guards to $d$ if those are taken
away to defend an attack. In particular, our aim is to assign
$g(p)=2^p+1$
private components for every $d \in D$ such that
each of those components is adjacent to $d$ (and can therefore be
used to supply a guard to $d$). Unfortunately, this might not always
be possible, e.g., if one vertex or a subset of vertices in $D$ does
not have sufficiently many adjacent components. To deal with this
situation we use a known generalisation of Hall sets as described below.

A subset $D'\subseteq D$ is an {\it $i$-Hall set} if $D'$ is adjacent to less than $i\cdot |D'|$  components in $\mathcal{D}$ for some $i\geq 1$. Let $D^*$ be the set obtained from $D$ as follows. In the beginning, we set $D^*=D$ and $G^*=G$. Then, at every step, if $D^*$ still contains a $g(p)$-Hall set w.r.t. $G^*$, let $D'\subseteq D^*$ be an inclusion-wise minimal $g(p)$-Hall set w.r.t. $G^*$. We then remove $D'$ from $D^*$ and $G^*$ and we remove components adjacent to $D'$ from $G^*$. We then repeat the process until $D^*$ no longer has any $g(p)$-Hall sets remaining w.r.t. $G^*$. Note that, at the end of this process, $D^*$ has no $g(p)$-Hall set and $D\setminus D^*$ has less than $|D\setminus D^*|g(p)$ adjacent components in $G$.

The following well-known generalisation of Hall's marriage theorem shows that every vertex in $D^*$ can be mapped to its own $g(p)$ components adjacent to it.

\begin{theorem}[{\cite{Di12}}]\label{thm:hmth}
A bipartite graph $G = (A \cup B, E)$ contains a subgraph where every vertex in $A$ has degree exactly $k$ (and every vertex in $B$ has degree at most 1) if and only if for all $X\subseteq A$, $|N(X)|\ge k|X|$.
\end{theorem}

Let us define a bipartite graph $H=(D^*\cup \mathcal{C},F)$ as follows: $\mathcal{C}$ is in a one-to-one correspondence with the components of $\mathcal{D}$ that are not adjacent with a vertex in $D\setminus D^*$. Denote by $v_C$ the vertex of $\mathcal{C}$ corresponding to the component $C$ of $\mathcal{D}$ that is not adjacent with a vertex in $D\setminus D^*$. Finally, for a vertex $d\in D^*$ and $v_C\in \mathcal{C}$, $dv_C\in F$ if and only if $d$ is adjacent with the component $C$ in $G$. Since $D^*$ has no $g(p)$-Hall set w.r.t. $G^*$, and equivalently $H$, by \Cref{thm:hmth}, there is a subgraph $H'$ of $H$ where every vertex in $D^*$ has degree exactly $g(p)$.

Note that for the vertices in $D\setminus D^*$, we will ensure that all their adjacent components are always fully occupied by guards, which is possible because there are not too many such components.

We are now ready to assign the components reserved for every vertex in $D$. That is, for every $d \in D^*$, we let $S(d)$ be the set of $g(p)$ private components assigned to $d$ by $H'$, i.e., $\{C\in \mathcal{D}~|~v_C\in N_{H'}(d)\}$. Moreover, note that the vertices in $D\setminus D^*$ altogether have at most $|D\setminus D^*|g(p)$ many components adjacent to them; for convenience, we set $S(d)$ to be the set of all components adjacent to $d$ for every $d \in D\setminus D^*$. Finally, let $\mathcal{M}=\mathcal{D}\setminus (\bigcup_{d\in D} S(d))$.


Given the assignment to private components, we are now ready to
construct the configuration $\confT^*$, which as we will show is a
winning configuration for the defender on $G$, as follows:
\begin{itemize}
\item[(a)] $D\subseteq \confT^*$;
\item[(b)] for every component $C\in \bigcup_{d\in D}S(d)$, it holds that $V(C)\subseteq \confT^*$;
\item[(c)] for every component $C\in \mathcal{M}$, it holds that
  $\confT^*\cap V(C)$ is a minimal defendable state in $G(C)$ using
  $lb(C)$ guards.
\end{itemize}
Note that the difference between $|\confT^*|$ and the lower bound $lb_D(G)=\sum_{C \in \mathcal{D}}lb(C)$ given in \Cref{obs:lowerbound} is at most $\ell+\sum_{C\in \bigcup S(d)} (|C|-lb(C))\leq \ell+\ell p(2^p+1)=\ell(1+p(2^p+1))=\Delta(\ell,p)$. Indeed: in $D$, $lb_D(G)$ counts 0 while $|\confT^*|$ counts $\ell$, for every component $C \in \bigcup_{d \in D} S(d)$, $lb_D(G)$ counts $lb(C)$ while $|\confT^*|$ counts $|C|$, and for every component $C\in \mathcal{M}$, $lb_D(G)$ and $|\confT^*|$ both count $lb(C)$.

\newcommand{\nosteps}{\textsf{NoS}}

It therefore merely remains to show that $\confT^*$ is winning for the
defender on $G$.
Let $C \in \mathcal{D}$, let $e$ be an edge incident to $V(C)$ in $G$, let $\confC_1
\subseteq V(C)$ be a defendable configuration for $C$ and let
$\confC_2=f_d(\confC_1\cup D \cup B,e)$ for some winning strategy $f_d$
for the defender on $G(C)$ starting from $\confC_1\cup D\cup B$.
Moreover, let
$\mathcal{P}$ be a set of paths that witnesses that $\confC_1\cup B\cup
D$ and $\confC_2$ are $e$-compatible.
Let $\mathcal{P}'$
be the set of paths obtained from $\mathcal{P}$ as follows. For every path
$P=(v_1,\dotsc,v_t) \in \mathcal{P}$ with $t\geq 2$, we do the
following. Without loss of generality, we can assume that $v_t\in
\confC_2\setminus (\confC_1\cup D\cup B)$. If $V(P)\subseteq V(C)$, we add $P$ to
$\mathcal{P}'$. Otherwise, $P$ contains a vertex in $D$ and we let $i$
be the maximum index such that $v_i \in V(D)$. If $v_t \in V(C)$, then
we add the path $(v_i,\dotsc, v_t)$ to $\mathcal{P}'$ and this
concludes the definition of $\mathcal{P}'$. Let $\confC_2'$
be the configuration obtained from $\confC_1\cup D\cup B$ using the
paths in $\mathcal{P}'$, then:
\begin{itemize}
\item $\confC_2\cap V(C) \subseteq
  \confC_2'\cap V(C)$ and therefore $\confC_2'\cap V(C)$ is defendable.
\item $\confC_1 \cup D^+$ and $\confC_2'\cap V(C)$ and are $e$-compatible
  in $G[C\cup D]-E(G[D])$ as
  witnessed by the paths in $\mathcal{P}'$, where $D^+$ is the subset
  of vertices in $D$ that are endpoints of the paths in
  $\mathcal{P}'$; note that $e$ is contained in a path of
  $\mathcal{P}'$ because $e$ has one endpoint $v$ in $V(C)$ with $v
  \notin \confC_1$.
\end{itemize}
\newcommand{\revc}{\textsf{revc}}
In the following, we denote by $\revc(\confC_1,f_d,e)$ the pair
$(D^+,\confC_2'\cap V(C))$.

We say that a configuration $\confC \subseteq V(C)$ is
\emph{reversible to a state $\confC_0\subseteq V(C)$} if there is a finite sequence
$((D_1,\confC_1),\dotsc,(D_x,\confC_x))$ of pairs
$(D_i,\confC_i)$ such that:
\begin{itemize}
\item $(D_1,\confC_1)=\revc(\confC_0,f_d,e)$ for some winning
  strategy $f_d$ for the defender on $G(C)$ and edge $e$ incident to
  $V(C)$ in $G$,
\item for every $i$ with $1\leq i < x$, it holds that $(D_{i+1},\confC_{i+1})=\revc(\confC_i,f_d,e)$ for some winning
  strategy $f_d$ for the defender on $G(C)$ and edge $e$ incident to
  $V(C)$ in $G$ and
\item $\confC_x=\confC$.
\end{itemize}
We say that a configuration $\confC$ is \emph{reversible} if it is
reversible to some minimal state, i.e., a defendable state using the
minimum number $lb(C)$ guards, and denote by $\nosteps(\confC)$ the
number of steps required by (length of) the witnessing sequence (that
brings the configuration back to a minimal state).

In the following, we show that the defender can win from $\confT^*$ by employing a strategy that ensures that the following properties hold for every configuration $\confT$ reached during the course of the game:
\begin{itemize}
\item[(1)] $D\subseteq \confT$;
\item[(2)] for every $d\in D^*$ and for every component $C\in S(d)$, $\confT$ contains all but at most one vertex of $C$;
\item[(3)] for every $d\in D\setminus D^*$ and for every component $C\in S(d)$, $V(C)\subseteq \confT$;
\item[(4)] for every component $C\in \mathcal{M}$, the state
  $\confT\cap V(C)$ is defendable;
\item[(5)] for every component $C\in \mathcal{M}$, the state
  $\confT\cap V(C)$ is reversible.
\item[(6)] $\sum_{C \in \mathcal{M}}\nosteps(\confT\cap V(C)) \leq g(p)$.
\end{itemize}
Note that the following property can be obtained from property (6) by
observing that the only components that can have received guards from
components in $S(d)$ through $d \in D^*$ are the components in
$\mathcal{M}$ that are not in a minimal state, and every step those
components are away (in terms of $\nosteps(*)$) from being in the minimal state can take at most
one guard from a component in $S(d)$. 
\begin{observation}\label{obs:ub-suf-full}
  Let $\confT$ be any configuration of $G$ satisfying (1)--(6), then
  for every $d \in D^*$ there are at least $|S(d)|-r$ full components
  in $S(d)$, where $r=\sum_{C \in \mathcal{M}}\nosteps(\confT\cap V(C))$.
\end{observation}

Let $\confT$ be a configuration reached during the course of the game
starting from $\confT^*$,
the rest of the proof shows that there exists a defender strategy $f_d$ such that, for every edge $e$,
$f_d(\confT,e)$
satisfies (1)--(6). Note that using induction and taking into account
that $\confT^*$ satisfies (1)--(6),
we can assume that $\confT$ satisfies (1)--(6).

{\bf Case 1: $e$ is in $G[D\cup C]$ for some component $C\in
  \mathcal{M}$.} By property $(4)$, there is a winning strategy $f_d'$
for the defender on $G(C)$ starting from $\confC\cup D\cup B$, where
$\confC=\confT\cap V(C)$.
Let
$(D^+,\confC')=\revc(\confC,f_d',e)$ and let $\mathcal{P}$ be a set of paths
witnessing that $\confC\cup D^+$ and $\confC'$ are
$e$-compatible. 

Suppose first that there is no component $C'\in\mathcal{M}\setminus\{C\}$ such
that $\nosteps(C')>0$. For every $d \in D^+$, let $C_d$ be a component
in $S(d)$ with $\confT\cap V(C_d)=V(C_d)$, i.e., a \emph{full} component, and let $n_d$ be a neighbor of $d$ in $C_d$. Note
that $C_d$ exists because of \Cref{obs:ub-suf-full}. Let
$\mathcal{P}'$ be the set of paths obtained from $\mathcal{P}$ after
adding the edge $dn_d$ to the unique path with endpoint $d\in
D^+$ in $\mathcal{P}$ for every $d \in D^+$.
Let $\confT'$ be the configuration obtained after applying the
paths in $\mathcal{P'}$ to $\confT$ and set
$f_d(\confT,e)=\confT'$. Then, $\confT'$ satisfies (1)--(6), because:
\begin{itemize}
\item[(1)] every guard removed from $d \in D^+$ was resupplied
  by $n_d$.
\item[(2)] we only removed at most one guard from full
  components in $S(d)$.
\item[(3)] $D^+\subseteq D^*$ since the component
  $C$ is in $\mathcal{M}$ and adjacent to every $d \in D^+$. 
\item[(4)] the only component in $\mathcal{M}$ that changed
  its state is $C$ and $C$ changed state into the defendable $\confC'$.
\item[(5)] As (4) and taking into account that $\confC$ was
  reversable.
\item[(6)] $C$ is the only component with $\nosteps(\confC)\neq 0$
  and $C$ has at most $2^p$ configurations and therefore
  $\nosteps(\confC)\leq 2^p\leq g(p)$.
\end{itemize}

So suppose now that there is a component $C' \in \mathcal{M}\setminus\{C\}$ such
that $\nosteps(C')>0$ and let
$((D_1^+,\confC_1),\dotsc,(D^+_x,\confC_x))$ be the witnessing
sequence for reversability for $C'$, which exists due to property (5).
Then, $\confC_{x-1}\cup D^+_x$ is
compatible with $\confC_x=\confT\cap V(C')$ and this is witnessed by the set
$\mathcal{P}'$ of paths. Our aim is to reverse the step from
$\confC_x$ to $\confC_{x-1}\cup D^+_x$ and use the
guards in $D^+_x$ to either replenish components in $S(d)$ for every $d \in
D^+_x\setminus D^+$ or to help with the attack on $C$ for every $d \in
D^+\cap D^+_x$.

For every $d \in D^+\setminus D^+_x$, let $C_d \in S(d)$ be a
full component, i.e., a component with $\confT\cap V(C_d)=V(C_d)$, which exists due to
\Cref{obs:ub-suf-full}.
Moreover, for every $d \in
D^+_x\setminus D^+$, let $C_d \in S(d)$ be a component with $|\confT\cap
V(C_d)|=|V(C_d)|-1$; note that such a component exists because
when the state of component $C'$ changed from $\confC_{x-1}$ to $\confC_{x}$ the
guards on $D^+_x$ were replenished using guards in $S(d)$ for $d \in
D^+_x$.
Let $\mathcal{P}''$ be obtained from
$\mathcal{P}$ and $\mathcal{P}'$ as follows. First, $\mathcal{P}''$
contains all paths in $\mathcal{P}\cup \mathcal{P}'$ that do not
contain any vertex in $D$. Moreover, for every vertex $d \in
D^+\setminus D^+_x$, let $n_d$ be the neighbor of $d$ in $C_d$ and let
$P \in
\mathcal{P}$ be the unique path containing $d$ as an endpoint. Then,
we add the path obtained from $P$ after adding the edge $dn_d$ to
$\mathcal{P}''$. For every $d \in D^+\cap D^+_x$, let $P\in
\mathcal{P}$ and $P'\in\mathcal{P}'$ be the unique paths in
$\mathcal{P}$ and $\mathcal{P}'$ with endpoint $d$. Then, we add the
concatenation of $P$ and $P'$ to $\mathcal{P}''$. Finally, let $d
\in D^+_x\setminus D^+$ and let $P \in \mathcal{P}'$ be the unique
path with endpoint $d$, let $n_d$ be a neighbor of $d$ in
$C_d$, and let $n_d'$ be the unique vertex in $V(C_d)\setminus
\confT$.
Then, we add the path obtained from $P$ after adding the path
$dn_dn_d'$ to $\mathcal{P}''$. Let $\confT'$ be the configuration
obtained from $\confT$ after applying the paths in
$\mathcal{P''}$ and set
$f_d(\confT,e)=\confT'$.
Note that $\confT'$ satisfies (1)--(6) because:
\begin{list}{}{%
  \settowidth{\labelwidth}{(4),(5)}%
  \setlength{\labelsep}{0.5em}%
  \setlength{\leftmargin}{\labelwidth}%
  \addtolength{\leftmargin}{\labelsep}%
  \renewcommand{\makelabel}[1]{#1\hfil}%
}
\item[(1)] every guard removed from $d \in D^+\setminus D^+_x$ was resupplied
  by $n_d$ and every guard removed from $D^+_x$ was resupplied
  from $C'$.
\item[(2)] we only removed at most one guard from full
  components in $S(d)$.
\item[(3)] $D^+\cup D^+_x\subseteq D^*$ since the components
  $C$ and $C'$ together are adjacent to every $d \in D^+\cup D^+_x$.
\item[(4),(5)] the only two components in $\mathcal{M}$ that changed
  their state are $C$ and $C'$, and both components were changed to a
  defendable and reversible configuration.
\item[(6)] the number of steps required for reversal increased by one for $C$
  but decreased by one for $C'$ and remained the same for all other components.
  Therefore, $R=R'$, where $R=\sum_{C \in \mathcal{M}}\nosteps(\confT\cap V(C))$
  and $R'=\sum_{C \in \mathcal{M}}\nosteps(\confT\cap V(C)')$.
\end{list}

{\bf Case 2: $e$ is in $G[D\cup C]$ for some component $C\in S(d)$,
  for some $d\in D^*$.}
In this case $C\setminus \confT=\{v\}$ for some $v \in e$. Let
$e=uv$. If $u \in V(C)$, then let $\confT'$ be the configuration
obtained from $\confT$ after moving the guard from $u$ to
$v$. Clearly, $\confT'$ satisfies (1)--(6) and we can set
$f_d(\confT,e)=\confT'$. Otherwise, $u \in D$. In this case let $C'
\in S(d)\setminus\{C\}$ be a full component, which exists due to
\Cref{obs:ub-suf-full},
and let $w$ be a neighbor of $u$ in
$C'$. Finally, let $\confT'$ be the configuration obtained from $\confT$
after moving the guard on $w$ to $v$ via the path $wuv$. Clearly,
$\confT'$ satisfies (1)--(6) and we can set $f_d(\confT,e)=\confT'$.

{\bf Case 3: $e$ is in $G[D\cup C]$ for some component $C\in S(d)$, for some $d\in D\setminus D^*$.}
By construction, both endpoints of $e$ belong to $\confT$. Thus, the defence is trivial.

We showed $\confT^*$ is part of a set of configurations which forms an eternal vertex cover of $G$. Thus, we obtain that $evc(G)\leq |\confT^*|=ub(G)$ and therefore.

\begin{lemma}\label{lem:up-low-dif}
Let $G=(V,E)$ be a graph and $D$ be a vertex set. Denote with $p$ the size of a largest component of $G-D$ and $\ell$ the size of $D$. Then, $evc(G)\leq ub(G)=lb_D(G)+\Delta(\ell,p)=lb_D(G)+\ell(1+p(2^p+1))$.
\end{lemma}

\begin{corollary}\label{cor:scd-up-low-uncon}
Let $G=(V,E)$ be a graph and $D$ be a vertex set. Denote with $p$
the size of a largest component of $G-D$ and $\ell$ the size of $D$.
Then the difference between the upper and lower bounds is at most
\[
  \beta(\ell,p):=\ell(1+p(2^p+1)).
\]
\end{corollary}

\subsection{An \textsf{XP} Algorithm}

We first observe that the standard techniques for the considered parameter
already suffice to place \EVC{} in \textsf{XP} when parameterized by \SCD{},
without any additional assumption on the components of $G-D$. The stronger
\textsf{FPT} result for self-sufficient instances is established in
\Cref{the:evc-fpt-scd}.

We begin with the terminology that is used throughout this section. Fix a graph
$G$, a set $D\subseteq V(G)$, and let every component of $G-D$ have size at most
$p$. The \emph{structural signature} of a component $C$ of $G-D$ records the
graph induced by $D\cup C$, with the vertices of $D$ distinguished. Equivalently,
two components $C$ and $C'$ have the same structural signature if there is an
isomorphism from $G[D\cup C]$ to $G[D\cup C']$ that fixes every vertex of $D$.
We denote the set of structural signatures by $\mathcal{S}_{G,D,p}$. Since
$|D|=\ell$ and every component has size at most $p$, the number of possible
structural signatures is bounded by a function $\tau(\ell,p)$.

For each structural signature $\kappa$, fix one representative component
$R_\kappa$ and, for every component $C$ with signature $\kappa$, fix an
isomorphism from $G[D\cup C]$ to $G[D\cup R_\kappa]$ that fixes $D$. A
\emph{state} of a component with signature $\kappa$ is then a subset of
$V(R_\kappa)$: a configuration on $C$ induces the corresponding subset of the
representative through the fixed isomorphism. Thus each structural signature has
at most $2^p$ states. This lets us record a configuration of the instance by recording, for each
signature and each state, how many components of that signature currently occur
in that state; the individual names of interchangeable components are irrelevant.

\iflong\begin{theorem}\fi
\ifshort\begin{theorem}[$\star$]\fi\label{the:evc-xp-scd}
\EVC{} is in \textsf{XP} parameterized by $\ell+p$, where $\ell$ is the size
of a \SCD{} $D$ of $G$ and $p$ is the size of the largest component of
$G-D$. In particular, \EVC{} is in \textsf{XP} parameterized by \SCD{}.
\end{theorem}

\iflong
\begin{proof}
Let $D$ be a \SCD{} of $G$ with $|D|=\ell$ such that every component of $G-D$
has size at most $p$, and let $n=|V(G)|$. We guess the value $k=evc(G)$; since
$k\leq n$, this incurs a factor of at most $n$ in the running time.

We associate with every configuration $\confT$ of order $k$ its
\emph{configuration vector} $\langle X, w_\confT\rangle$, where
$X=\confT\cap D$ and, for every pair $(\kappa,\mathfrak{s})$ consisting of a
structural signature and a state of that signature, the entry
$w_\confT(\kappa,\mathfrak{s})$ records the number of components of structural
signature $\kappa$ that are in state $\mathfrak{s}$ under $\confT$. The number
of such pairs is at most $d(\ell,p):=\tau(\ell,p)\cdot 2^p$, each entry of
$w_\confT$ lies in $\{0,\dots,n\}$, and $X$ ranges over the $2^\ell$ subsets of
$D$. Hence the number of distinct configuration vectors is at most
$2^\ell\,(n+1)^{d(\ell,p)}=n^{\mathcal{O}(d(\ell,p))}$.

Because components of the same structural signature are interchangeable, any two
configurations with the same configuration vector are related by an
automorphism of $G$ that fixes $D$ pointwise; consequently one is winning for
the defender if and only if the other is. Moreover, for a fixed edge $e$, whether
two configurations are $e$-compatible depends only on their configuration
vectors and on the structural signatures and states of the components incident
with $e$. We may therefore work with the quotient
$\mathfrak{C}^*_k$ of the configuration graph $\mathfrak{C}_k(G)$ obtained by
identifying configurations that share the same configuration vector. 
Therefore, it follow from \Cref{obs:conf-equivalence} that we can construct $\mathfrak{C}^*_k$ and use it to verify whether EVC is at most $k$ in polynomial-time w.r.t. the number of vertices $n^{\mathcal{O}(d(\ell,p))}$ of $\mathfrak{C}^*_k$.


Finally, using the $(k+1)$-approximation algorithm for \SCD~\cite{DEGKO21}, we
can compute a deletion set $D$ for which both $\ell$ and $p$ are bounded by a
function of $\scd(G)$. Hence the algorithm runs in time $n^{h(\scd(G))}$ for a
computable function $h$, that is, \EVC{} is in \textsf{XP} parameterized by
\SCD{}.
\end{proof}
\fi

\subsection{\textsf{FPT} and Additive Approximation Algorithms}

The main result of this section is the following.

\iflong\begin{theorem}\fi
\ifshort\begin{theorem}[$\star$]\fi
\label{the:evc-fpt-scd}
\EVC{} is fixed-parameter tractable parameterized by $\ell + p$, where $\ell$ is the size of a \SCD{} $D$ of $G$ and $p$ is the size of the largest component of $G - D$, provided that all components of $G - D$ are self-sufficient.
\end{theorem}

\iflong
Towards the algorithm, we begin by guessing the size of $evc(G)$, and let this guess be denoted $k$ and fixed for the rest of this discussion. By~\Cref{cor:scd-up-low-uncon}, the possible values of $k$ lie in an interval of length bounded by $\beta(\ell,p)$, so this guessing step only adds a factor depending on $\ell$ and $p$ to the running time.
Let $G$ be a graph with a \SCD{} $D$ of size $\ell$ and let $p$ be the size of a largest component of $G-D$. We use the structural signatures introduced at the start of this section. Thus $\mathcal{S}_{G,D,p}$ denotes the set of signatures of components of $G-D$, two components have the same signature precisely when their induced graphs together with $D$ are isomorphic via an isomorphism that fixes $D$ pointwise, and $|\mathcal{S}_{G,D,p}|\leq \tau(\ell,p)$ for some function $\tau$ depending only on $\ell$ and $p$. We also keep the fixed representative and fixed isomorphisms chosen there, so that states of components with the same signature can be compared as subsets of the same representative component.

Set $\theta(\ell,p):=\beta(\ell,p)+2\ell+1$ and $h(\ell,p):=2^p\theta(\ell,p)$. The choice of these thresholds will be apparent in due course.


\begin{definition}[State]\label{def:state}
Let $G$ be a graph with a \SCD{} $D$. For each structural signature $\kappa \in \mathcal{S}_{G,D,p}$, we denote the set of possible states of $\kappa$ by $2^\kappa$. This is the collection of all subsets of the representative component for $\kappa$.
\end{definition}

\begin{definition}[Minimal State]
Let $G$ be a graph with a \SCD{} $D$ and let $C$ be a component of $G
- D$. Let $\confT$ be a configuration of $G$. The configuration
$\confT$ is said to be minimal for $C$ if $|\confT \cap V(C)|=lb(C)$, where $lb(C)$ is the lower-bound contribution defined in \Cref{obs:lowerbound}.
The state associated with a minimal configuration is called a minimal state.
\end{definition}

We now introduce terminology to address components that are not minimally defended in some configuration.

\begin{definition}
  Let $G$ be a graph with a \SCD{} $D$ and let $\confT$ be a configuration of $G$. We say that a component $C$ in $G - D$ is \emph{heavy} with respect to $\confT$ if $|\confT\cap V(C)| > lb(C)$.
\end{definition}

We now observe that we cannot have too many heavy components when defending with $k$ guards.

\begin{observation}\label{ob:heavy-components}
Let $G$ be a graph with a \SCD{} $D$ of size at most $\ell$ and let
$p$ be the size of a largest component of $G-D$. Further, assume that
the defender has a winning strategy $f_d$ for $G$ starting from a configuration $\confT^*$ of size $k\leq ub(G)$.
Let $\confT$ be a configuration reached during the course of the game played based on $f_d$. Then $\confT$ has at most $\beta(\ell,p)$ heavy components. 
\end{observation}

\begin{proof}
Since $f_d$ is a winning strategy, we know that it deploys at least $lb_D(G) := \sum_{C \in \mathcal{D}} lb(C)$ guards. On the other hand, since $k \leqslant ub(G)$, we know that $f_d$ uses at most $ub(G)$ guards. If there are more than $\beta(\ell,p)$ components that use more than $lb(C)$ guards, then the total number of guards exceeds $lb_D(G)+\beta(\ell,p)\geq ub(G)$, a contradiction.
\end{proof}

It is also useful to know that the number of components that participate in guard exchanges involving vertices outside the component is bounded.

\begin{definition}
Let $G$ be a graph and let $D$ be a \SCD{} of $G$. Suppose
$\mathcal{P}$ is an applicable family of vertex-disjoint oriented paths. An edge of a path in $\mathcal{P}$ is said to be a
\emph{cross-edge} if exactly one of its endpoints lies in $D$, i.e., it
corresponds to a guard move between $D$ and some component of $G-D$.
All other path edges are called \emph{internal} edges.
\end{definition}

The following observation follows from the fact that every vertex in $D$ can receive a guard at most once and lose a guard at most once in one round of defense.

\begin{observation}\label{ob:cross-edges}
Let $G$ be a graph and let $D$ be a \SCD{} of $G$. The number of cross-edges in any applicable path family is at most $2\ell$.
\end{observation}

Now we are ready to present the main intuition for the algorithm: we know that in principle the number of components that participate in external guard exchanges is bounded, as are the number of heavy components. So the vast majority of the components are untouched and can freely switch between minimal configurations by self-sufficiency. This enables us to propose the following algorithm. For each structural signature $\kappa$, delete all but $h(\ell,p)=2^p\theta(\ell,p)$ components in $G-D$ that have signature $\kappa$. Suppose the number of deleted components is $w_\kappa$. Then we reduce the guard budget by $w_\kappa\cdot lb(\kappa)$, where $lb(\kappa)$ denotes $lb(C)$ for any component $C$ with signature $\kappa$. Repeat this across all structural signatures, and this will leave us with an instance whose size is bounded by $\ell + \tau(\ell,p)\cdot h(\ell,p)$, with a guard budget of $k-\sum_{\kappa \in \mathcal{S}_{G,D,p}} w_\kappa \cdot lb(\kappa)$. The key idea here is that the components that were removed by the algorithm can be predictably handled with $lb(\cdot)$ extra guards, and this is because of \Cref{ob:heavy-components} and \ref{ob:cross-edges} about the number of cross edges and heavy components being bounded. We now formalize these ideas.


When we set out to show the equivalence between the original graph and the pruned one, we see that while the deleted components can be safely assumed to not be involved in any external guard exchanges, in the original graph the defense may require for their state to be internally reconfigured. To this end, we will require that components that are only internally affected can freely switch between minimum extendable positions, and this motivates the definition of self-sufficiency. A self-sufficient component is one where we are able to switch between any two minimum extendable positions without ``outside help''.

\begin{definition}[Self-Sufficient Components]\label{def:self-sufficient}
Let $G$ be a graph with a \SCD{} $D$ and let $C$ be a component of $G - D$. A (partial) configuration $P_C \subseteq V(C)$ for $C$ is \emph{extendable} if there is a configuration $\confT$ for the defender on $G$ such that $V(C) \cap \confT = P_C$. We say that $P_C$ is \emph{minimum} if it has minimum size among all extendable positions of $C$. Then $C$ is \emph{self-sufficient} if the defender can move in one round between any two minimum extendable positions of $C$ without exchanging any guard with the deletion set, provided that there is no attack on $C$.
\end{definition}

Finally, a nuance worth noting here is that we start with the input graph $G$ and a \SCD{}, and we do not have an explicit configuration of guards at hand. The threshold $h(\ell,p)$ is chosen large enough to reserve, for every possible state of a fixed structural signature, enough representative components to simulate the bounded number of components that can participate in exchanges with $D$. If the input has fewer than this many components of a signature, then that signature is already bounded and no deletion is needed. To argue the equivalence of the original graph and the pruned graph, we will need to go back and forth between their configurations. To this end, we introduce the signature of a \emph{configuration} with respect to $D$.

\begin{definition}[Configuration Signature]
Let $G$ be a graph and let $D$ be a \SCD{} of $G$. Given a
configuration $\confT \in V(\mathfrak{C}_k(G))$, the \emph{configuration signature}
of $\confT$ with respect to a deletion set $D$, denoted $v_\confT(D)$,
is a pair $\langle X, w_\confT \rangle$, where $X = D\cap \confT$ and $w_\confT$ is a vector indexed by $\{(\kappa,\mathfrak{s}) \mid \kappa \in \mathcal{S}_{G,D,p}, \mathfrak{s} \in 2^\kappa\}$. The subset $X$ specifies the set of vertices in $D$ that are occupied by \confT, and the entry $w_\confT(\kappa, \mathfrak{s})$ gives the number of components in $G- D$ with structural signature $\kappa$ that are in state $\mathfrak{s}$ in configuration $\confT$.
\end{definition}

\begin{lemma}
  \label{thm:scd}
  Let $G$ be a graph with a \SCD{} $D$ of size at most $\ell$ such that all components in $G - D$ are self-sufficient and the largest component of $G-D$ has size at most $p$. If there are more than $h(\ell,p)$ components in $G-D$ that have the same structural signature $\kappa$, delete all but $h(\ell,p)$ of them from $G$. For each signature $\kappa$, let $w_\kappa$ denote the number of components of signature $\kappa$ that were deleted by this process. Further, let $H$ be the graph obtained from $G$ after performing these deletions. Let
  \[
    \eta(G,D) := \sum_{\kappa \in \mathcal{S}_{G,D,p}} w_\kappa \cdot lb(\kappa),
  \]
  where $lb(\kappa)$ denotes $lb(C)$ for any component $C$ with structural signature $\kappa$.

  Then $evc(G) = evc(H) + \eta(G,D)$.
  \end{lemma}
  Note that $H$, as defined in the statement of~\Cref{thm:scd}, is a subgraph of $G$. Also observe that this lemma implies a \textsf{FPT} algorithm in $\ell$ and $p$ since the size of the graph $H$, which can be obtained in \textsf{FPT} time from the graph $G$, is bounded by $\ell+\tau(\ell,p) \cdot h(\ell,p)$, and here EVC can be solved by an exact algorithm.

  Before proving~\Cref{thm:scd}, we introduce some definitions that allow us to go back and forth between configurations in the graphs $G$ and $H$. We first introduce some terminology to capture which signatures were ``affected'' by the component deletions. In particular, we want to focus on signatures $\kappa$ which are such that there is at least one component in $G - V(H)$ that has signature $\kappa$.

  For a structural signature $\kappa$, let $\mathcal{E}^G_\kappa$ and
  $\mathcal{E}^H_\kappa$ denote the components of $G-D$ and $H-D$,
  respectively, that have signature $\kappa$.

  \begin{definition}\label{def:affected}
  Let $G$ be a graph with a \SCD{} $D$, and let $\confT$ be a
  configuration of guards on $G$. Let $\kappa$ be a signature and let
  $s \in 2^\kappa$ be a state. We say that $(\kappa,s)$ is
  \emph{crowded} in $\confT$ if
  $w_\confT(\kappa,s)>\theta(\ell,p)$. Further, $\kappa$ is said to be \emph{affected} if there is some state $s \in 2^\kappa$ such that $(\kappa,s)$ is crowded in $\confT$. Since there are at most $\beta(\ell,p)$ heavy components and $\theta(\ell,p)>\beta(\ell,p)$, every crowded state is a minimal state.
  \end{definition}

  Next, we associate a single crowded state with every affected signature. This is not strictly speaking needed for our arguments but it simplifies our discussion.

  \begin{definition}
  Let $\succ_\kappa$ be an arbitrary but fixed ordering on $2^\kappa$. For an affected signature $\kappa$, let $s$ be the smallest state such that $(\kappa,s)$ is crowded in $\confT$ with respect to $\succ_\kappa$. We say that $s$ is the \emph{key state} of the signature $\kappa$ in $\confT$.  Note that the key state of a signature is well-defined only for affected signatures, and by 
  \Cref{def:affected}, it is always minimal.
  \end{definition}

  We are now ready to define the projection of a configuration of $G$ in $H$ and the extension of a configuration of $H$ in $G$.

  \begin{definition}
  Let $\confT$ be a configuration of guards in $G$ with a \SCD{} $D$.
  First cap the number of components of each signature in each state by
  setting
  \[
    \widehat w_\confT(\kappa,s)=
      \min\{w_\confT(\kappa,s),\theta(\ell,p)\}.
  \]
  For every structural signature $\kappa$, define
  \[
    q_\kappa(\confT)=|\mathcal{E}^H_\kappa|-
      \sum_{s\in 2^\kappa}\widehat w_\confT(\kappa,s).
  \]
  This number is non-negative by the definition of $h(\ell,p)$. If
  $q_\kappa(\confT)>0$, then $\kappa$ is affected, so it has a key state
  in $\confT$. The \emph{projection} of $\confT$ in $H$ is any
  configuration $\confR$ whose configuration signature is
  $\langle \confT\cap D,w_\confR\rangle$, where

  \begin{equation*}
    w_\confR(\kappa,s) =
    \begin{cases}
      \widehat w_\confT(\kappa,s)+q_\kappa(\confT)
        & \text{if } \kappa \text{ is affected and } s \text{ is its key state},\\
      \widehat w_\confT(\kappa,s)
        & \text{otherwise.}
    \end{cases}
  \end{equation*}
  If $\kappa$ is not affected, then $q_\kappa(\confT)=0$.
  \end{definition}

  \begin{definition}
  For every structural signature $\kappa$, fix one minimum state
  $s^0_\kappa$. Let $\confR$ be a configuration of guards in $H$ with a
  \SCD{} $D$. For each $\kappa$, choose a minimum state
  $a_\kappa(\confR)$ as follows: if some minimum state $s$ satisfies
  $w_\confR(\kappa,s)\geq \theta(\ell,p)$, then take the smallest such
  state with respect to $\succ_\kappa$; otherwise take
  $a_\kappa(\confR)=s^0_\kappa$. The \emph{extension} of $\confR$ to
  $G$ is any configuration $\confT$ whose configuration signature is
  $\langle \confR\cap D,w_\confT\rangle$, where

  \begin{equation*}
    w_\confT(\kappa,s) =
    \begin{cases}
      w_\confR(\kappa,s)+w_\kappa & \text{if } s=a_\kappa(\confR),\\
      w_\confR(\kappa,s) & \text{otherwise.}
    \end{cases}
  \end{equation*}
  Thus extension adds the deleted components back in a minimum state and increases the number of guards by exactly $\eta(G,D)$; recall that $w_\kappa$ denotes the number of components of signature $\kappa$ that were deleted by the process described in~\cref{thm:scd}.
  \end{definition}

We are now ready to prove~\Cref{thm:scd}.

\begin{proof}[Proof of \Cref{thm:scd}] We first show that $evc(G) \leqslant evc(H) + \eta(G,D)$. For brevity, we denote $\eta := \eta(G,D)$ and $\gamma := evc(H)$. Towards the upper bound, we have to prove that if there exists a winning defender strategy $f_d$ for $H$ that uses at most $\gamma$ guards, then there is a winning defender strategy $f_d^\prime$ in $G$ that uses at most $\gamma + \eta$ guards. For any $uv \in E(H)$, define $f_d'(\confT,uv)$ as the extension of $f_d(\confR,uv)$ to $G$, where $\confR$ is the projection of $\confT$ in $H$.

  It is easy to verify that $f_d^\prime$ uses $\gamma + \eta$ guards, because the extension places the deleted components in minimum states and therefore adds exactly $\eta$ guards. Also note that if $\confR'$ and $\confR$ are $uv$-compatible in $H$, then $\confT'$ and $\confT$ are $uv$-compatible in $G$, where $\confT$ is the extension of $\confR$ and $\confT^\prime$ is the extension of $\confR^\prime$. This implies that $f_d^\prime$ is a winning defender strategy in $G$ when an edge from $H$ is attacked in $G$.

  Now suppose an edge $uv \in E(G) \setminus E(H)$ is attacked when $G$ is in configuration $\confT$. Note that there exists a component $C$ in $G-D$ such that $uv \cap E(C) \neq \emptyset$, and $C$ is not a component in $H-D$. Let $\kappa$ denote the signature of $C$. We know that there exists some component $C^\star$ in $H-D$ whose signature is also $\kappa$. Let $\phi_C: C \rightarrow C^\star$ be an isomorphism between $C$ and $C^\star$ and let $x := \phi_C(u), y := \phi_C(v)$. Now place $H$ in the configuration $\confR$ that is the projection of $\confT$ in $H$ and attack the edge $xy$. The defense $f_d(\confR,xy)$ leads us to a configuration $\confR'$ in $H$, and suppose $\mathcal{P}$ is a list of directed edges which gives us a recipe for transforming $\confR$ into $\confR'$ in $H$. Replace every occurrence of $v \in C^\star$ with $\phi_C^{-1}(v)$ in $\mathcal{P}$ to obtain a list of directed edges $\mathcal{P}'$: this gives us a recipe for transforming $\confT$ into some $\confT'$ in $G$ while defending the edge $uv$.

  Now we want to ensure that the projection of $\confT^\prime$ is in fact $\confR^\prime$. Observe that $\confT^\prime$ mimics the defense $f_d$ exactly, other than swapping the roles of $C$ and $C^\star$. If there is a mismatch in the projection of $\confT^\prime$ and $\confR^\prime$, then it the mismatch must occur only at a minimal state, since non-minimal states correspond to heavy components, and since those are bounded, they are exactly preserved by the projections. It can be verified that any mismatch can be fixed by moving $C^\star$ (an untouched component) to the appropriate minimal state, which is possible to do by our assumption of self-sufficiency.

   We now show that $evc(G) \geqslant evc(H) + \eta(G,D)$, which is to say that if we have a strategy for defending attacks perpetually in $G$ using $\lambda := evc(G)$ guards, we can also defend attacks in $H$ with a reduced budget of $\lambda - \eta$ guards. Let $f_d$ be such a strategy for $G$. Define $f_d^\prime(\confR,uv)$ as the projection of $f_d(\confT,uv)$ in $H$, where $\confT$ is the extension of $\confR$ to $G$. It can be verified that the transition from $\confT$ to $f_d(\confT,uv)$ can be carried out entirely in $H$: indeed every deleted component $C$ that is involved in the transition can be injectively mapped to an ``inactive'' component of $H-D$ that shares the same signature. It is also easy to see that if $\confT'$ and $\confT$ are $uv$-compatible in $G$, then $\confR'$ and $\confR$ are $uv$-compatible in $H$. Finally, note that this is a well-defined strategy since $uv \in E(H)$ implies that $uv \in E(G)$, and uses the claimed number of guards as well, since the projection eliminates exactly $\eta(G,D)$ guards. This concludes our argument.
  \end{proof}

\begin{proof}[Proof of \Cref{the:evc-fpt-scd}]
Let $G$ be a graph with a \SCD{} $D$ of size $\ell$ such that all
components of $G - D$ are self-sufficient and the largest component
has size at most $p$. By \Cref{cor:scd-up-low-uncon}, we can guess the
value of $evc(G)$ at an additional factor depending only on $\ell$ and
$p$ in the running time. We then apply \Cref{thm:scd} to obtain, in
polynomial time, a graph $H$ whose size is bounded by
$\ell+\tau(\ell,p) \cdot h(\ell,p)$ --- a function of $\ell$ and $p$ alone
--- together with the value $\eta(G,D)$ such that
$evc(G) = evc(H) + \eta(G,D)$. We solve EVC exactly on $H$ in time
depending only on $|V(H)|$, which is bounded by a function of $\ell$
and $p$.
\end{proof}

We are now ready to provide our polynomial-time approximation.
The following corollary now follows because of
\Cref{lem:up-low-dif} together with the known
$(k+1)$-approximation algorithm for SCD~\cite{DEGKO21}. In particular,
if $k=\scd(G)$, then the approximation algorithm returns a deletion set
and component-size bound with $\ell+p\leq K:=k(k+1)$. Since both
$\ell$ and $p$ are at most $K$, the difference between the upper bound
and the lower bound given in \Cref{lem:up-low-dif} is at most
$K(1+K(2^K+1))$:
\begin{corollary}
  There is a polynomial-time algorithm that computes the EVC number of
  any graph $G$ to within an additive error of at most
  $K(1+K(2^K+1))$, where $K=\scd(G)(\scd(G)+1)$.
\end{corollary}
\fi

\section{Conclusion}
We provide an \textsf{FPT}-algorithm for eternal
vertex cover parameterized by \CVD{} and an \textsf{XP}-algorithm (which can be
improved to an \textsf{FPT}-algorithm for the case of self-sufficient
components) for the problem parameterized by \SCD{}. Towards this aim, we
also obtain upper and lower bounds for EVC based on \CVD{} and \SCD{} that
turn out to be crucial to obtain our \textsf{FPT}-results and also allow us to
obtain polynomial-time approximation algorithms for EVC with an
additive error that only depends on \CVD{} and \SCD{}.
The main open
question from our analysis is whether EVC is \textsf{FPT} parameterized by \SCD{},
i.e., whether the condition on the components to be self-sufficient
can be removed. Apart from this, the study of EVC with respect to
structural parameters is still in its infancy (and turns out to be
very challenging due to the dynamic nature of the parameter) and there
therefore remain many open questions. For instance, what is the
parameterized complexity of EVC parameterized by any of the following
parameters: split vertex deletion, deletion to paths (even edge
deletion), feedback vertex/edge set, treedepth, and treewidth. All of these questions seem to be
highly challenging, but we hope that the insights obtained here will
be helpful to resolve these questions in the future.


\newpage
\bibliography{references}

\end{document}